\documentclass[12pt, a4paper]{article}
\usepackage{arxiv}

\usepackage[utf8]{inputenc} % allow utf-8 input
\usepackage[T1]{fontenc}    % use 8-bit T1 fonts
\usepackage[hidelinks]{hyperref}       % hyperlinks
\usepackage{url}            % simple URL typesetting
\usepackage{booktabs}       % professional-quality tables
\usepackage{amsfonts}       % blackboard math symbols
\usepackage{nicefrac}       % compact symbols for 1/2, etc.
\usepackage{microtype}      % microtypography
\usepackage{lipsum}
\usepackage{graphicx}
\usepackage{orcidlink}
\usepackage{amsmath,amssymb,amsthm}
\usepackage{multirow}
\usepackage{natbib}
\usepackage{float}
\usepackage{subcaption}
\usepackage{makecell}
\usepackage{adjustbox}

\usepackage{algorithm}
\usepackage{algpseudocode}

\usepackage{amsthm}

\newtheorem{proposition}{Proposition}

\theoremstyle{remark}

\newcommand{\tmax}{t_{\max}}

\newcommand{\phat}{\widehat{p}}
\newcommand{\pKM}{\widehat{p}_{KM}}

\newcommand{\Tparam}{T_n^{\mathrm{Par}}}

\title{A parametric framework for assessing and estimating sufficient follow-up time in cure models}

\author{
  Luiz Silva-Resende \orcidlink{0000-0003-4697-6503} \\
  Department of Statistics \\
  Federal University of S{\~a}o Carlos \\
  S{\~a}o Carlos, Brazil \\
  \texttt{luizfsr@ufscar.br} \\
  \phantom{Institute of Mathematics and Computer Sciences}
\And
  Dionisio Alves-Neto \orcidlink{0000-0002-2086-4936} \\
  Institute of Mathematics and Computer Sciences\\
  University of S{\~a}o Paulo\\
  S{\~a}o Carlos, Brazil \\
  \texttt{dionisioneto@usp.br} \\
\And
  Danilo Alvares \orcidlink{0000-0003-3764-0397} \\
  MRC Biostatistics Unit\\
  University of Cambridge\\
  Cambridge, UK \\
  \texttt{danilo.alvares@mrc-bsu.cam.ac.uk} \\
  \phantom{Institute of Mathematics and Computer Sciences}
\And
  Vera Tomazella \orcidlink{0000-0002-6780-2089} \\
  Department of Statistics \\
  Federal University of S{\~a}o Carlos \\
  S{\~a}o Carlos, Brazil \\
  \texttt{vera@ufscar.br} \\
  \phantom{Institute of Mathematics and Computer Sciences}
}

\begin{document}

\maketitle

\begin{abstract}
Reliable estimation of cure fractions depends critically on adequate follow-up. Classical procedures for assessing follow-up sufficiency are primarily inferential, whereas a separate body of literature estimates time to cure using excess-hazard, conditional cure-probability, or residual-survival definitions. These two approaches remain largely disconnected: assessing the adequacy of observed follow-up does not generally quantify the additional follow-up required under an explicit tolerance, whereas estimating time to cure does not itself provide a formal inferential assessment of the available follow-up. We develop a unified parametric framework that combines the Parametric Follow-up Sufficiency Test, which compares the terminal Kaplan-Meier estimate with the cure fraction estimated from a parametric mixture cure model, with two complementary tolerance-based formulations: the Plateau Distance Criterion on the population survival scale and the Residual Survival Criterion on the susceptible survival scale. We derive closed-form expressions for both criteria under the Weibull mixture cure model and establish that they yield identical minimum follow-up times under an appropriate transformation of their tolerance parameters. We evaluate the finite-sample performance of the framework through Monte Carlo simulations across sample sizes, cure fractions, and administrative censoring scenarios. Applications to prostate cancer and triple-negative breast cancer data illustrate how the framework assesses available follow-up, estimates additional observation time under prespecified tolerances, and identifies settings in which follow-up is already adequate.
\end{abstract}

\section{Introduction} 
\label{sec:intro}

Cure rate models play a fundamental role in survival analysis when a proportion of individuals is expected to remain permanently free of the event of interest. Applications are particularly common in oncology, epidemiology, and reliability studies, in which advances in treatment or intervention may produce a subgroup of long-term survivors whose mortality eventually approaches that of the general population \citep{yilmaz2013insights, dal2019prognosis, johnen2021modeling}. Since the pioneering work of \cite{Boag1949} and \cite{Berkson1952}, cure models have been widely used to estimate the proportion of individuals considered cured and to investigate factors associated with long-term survival.

A fundamental requirement for reliable estimation of the cure fraction is a sufficiently long follow-up period \citep{tai2005minimum}. When follow-up is inadequate, a substantial proportion of susceptible individuals may remain event-free simply because they have not been observed long enough. In such settings, the estimated cure fraction may be biased, often upward, leading to misleading conclusions about the cure fraction, long-term prognosis, and treatment effectiveness.

Classical methods for assessing follow-up sufficiency have been formulated mainly as hypothesis-testing procedures. The seminal work of \citet{mallerandZhou1994testing} established support-based conditions for sufficient follow-up, whereas subsequent developments introduced alternative nonparametric and bootstrap procedures, including the recent test of \citet{xie2024testing}. These methods provide formal evidence about whether the available follow-up is adequate, but they do not directly calculate the additional observation time required when follow-up is judged insufficient.

A related but distinct literature has focused on estimating time to cure. \citet{chauvenet2009prevalence} defined a practical cure threshold as the time at which fewer than \(10\%\) of the patients expected to die from cancer remained alive. \citet{dal2014long} considered the time required for conditional relative survival over prespecified future windows to exceed clinically selected thresholds. \citet{boussari2018new} defined time to cure as the earliest time at which the conditional probability that an individual who remains alive belongs to the cured group reaches \(0.95\); \citet{romain2019time} subsequently applied this definition in a large population-based study of solid cancers. \citet{jakobsen2020estimating} placed these and related definitions within a broader framework in which a cure point is obtained by comparing a dynamic mortality measure with a prespecified margin of clinical relevance. \citet{boussari2021modeling} proposed an excess-hazard cure model in which the time to null excess hazard is a covariate-dependent parameter. More recently, \citet{dimari2025flexible} considered the time after which only a prespecified negligible proportion of fatal cases remains alive. From a different perspective, \citet{selukar2023receus} proposed the RECeUS diagnostic, which evaluates cure-model appropriateness through the estimated proportion of uncured subjects remaining at the administrative censoring time.

These contributions demonstrate that follow-up adequacy and time to cure can be studied through several complementary targets. Susceptible-survival thresholds were already considered by \citet{chauvenet2009prevalence}, and \citet{boussari2018new} explicitly expressed the conditional probability of cure in terms of the cure fraction and the survival function of uncured patients. However, these approaches do not jointly connect a formal inferential assessment of observed follow-up with a direct calculation of the additional observation time required under interpretable tolerance levels. The contribution of the present framework therefore lies not in introducing the underlying inverse-threshold principle, but in integrating it with a formal follow-up sufficiency test, expressing the resulting requirements on both the population and susceptible survival scales, establishing transformations between these scales, and quantifying the additional follow-up required beyond the observed study endpoint.

From a practical perspective, this gap is particularly relevant in large population-based cancer registries and longitudinal cohort studies. Researchers and health authorities frequently need to determine not only whether the current follow-up is adequate, but also how many additional years of observation are required before cure estimates can be interpreted under a prespecified practical tolerance. Such information is relevant for study design, resource allocation, and the interpretation of long-term survival outcomes.

Existing methods answer different questions about follow-up adequacy, but investigators still lack a unified statistical framework capable of simultaneously assessing whether follow-up is sufficient and estimating how much additional follow-up would be required under explicit tolerance levels. This methodological gap motivates the framework proposed in this paper.

To connect these inferential and estimation perspectives, we develop a parametric framework with three components. First, we introduce the Parametric Follow-up Sufficiency Test (PFST) for assessing whether the available follow-up is adequate. Second, we express the minimum follow-up requirement on two complementary scales through the Plateau Distance Criterion (PDC) and the Residual Survival Criterion (RSC). Third, we establish the exact mathematical relationship between the PDC and the RSC under an appropriate transformation of their tolerance parameters.

Although the proposed framework applies to any parametric cure model with an invertible susceptible survival function, we focus on the Weibull mixture cure model because it combines modeling flexibility with closed-form expressions for the criteria. Unlike existing nonparametric approaches, the framework uses the extrapolation capability of the fitted parametric model to project survival beyond the maximum observed follow-up time, allowing straightforward calculation of the minimum sufficient follow-up time for user-specified tolerance levels.

The main methodological contributions of this work are as follows:

\begin{enumerate}
\item It introduces the \textit{Parametric Follow-up Sufficiency Test} (PFST), which formally compares the terminal Kaplan--Meier estimate with the cure fraction estimated from a parametric mixture cure model. Inference is developed through both a centered nonparametric bootstrap procedure and an influence-function-based asymptotic approximation.

\item It provides two complementary criteria for estimating the minimum sufficient follow-up time: the \textit{Plateau Distance Criterion} (PDC), defined on the population survival scale, and the \textit{Residual Survival Criterion} (RSC), defined on the susceptible survival scale.

\item It derives closed-form expressions for both criteria under the Weibull mixture cure model, allowing straightforward estimation of the minimum follow-up time required to satisfy prespecified tolerance levels.

\item It establishes the mathematical relationship between the PDC and the RSC, demonstrating that the two criteria are mathematically equivalent under an appropriate transformation of their tolerance parameters while providing complementary scientific interpretations.
\end{enumerate}
% The main methodological contributions are fourfold. First, the PFST provides a formal comparison between the terminal Kaplan-Meier estimate and the cure fraction estimated from a parametric mixture cure model, with inference based on both a centered bootstrap and an influence-function approximation. Second, minimum follow-up requirements are expressed on the population and susceptible survival scales. Third, closed-form expressions for the proposed estimation criteria are derived under the Weibull mixture cure model. Finally, the mathematical relationship between the PDC and the RSC is established and their practical interpretation is demonstrated through simulations and cancer applications.

The proposed methodology is illustrated using a large population-based cohort of patients with prostate cancer. The application demonstrates how the framework can estimate the minimum follow-up time required to attain prespecified levels of proximity to the cure plateau while simultaneously assessing whether the available follow-up is adequate. The second application, based on triple-negative breast cancer data, illustrates a contrasting setting in which the available follow-up is already sufficient.

The remainder of the paper is organized as follows. Section~2 reviews the concept of follow-up sufficiency in cure models and discusses its theoretical foundation. Section~3 introduces the PFST, develops its centered bootstrap and influence-function-based inferential implementations, and presents the PDC and RSC under the Weibull mixture cure model. Section~4 presents a simulation study evaluating the finite-sample performance of the proposed procedures. Section~5 illustrates the methodology using two datasets on prostate cancer and triple-negative breast cancer. Finally, Section~6 concludes with a discussion of the main findings, practical implications, and directions for future research.

\section{Assessing follow-up sufficiency in cure rate models}

Reliable estimation of the cure fraction is a primary objective of cure rate models. Such estimation, however, is meaningful only when the available follow-up is sufficiently long to distinguish individuals who are truly cured from those who remain susceptible but have not yet experienced the event.

When follow-up is prematurely terminated, some susceptible individuals may still be event-free at the end of the study and therefore become indistinguishable from cured subjects. As a result, the cure fraction may be overestimated, leading to potentially misleading conclusions regarding long-term prognosis and treatment effectiveness. Consequently, assessing follow-up sufficiency is a prerequisite for the reliable interpretation of cure fraction estimates.

This issue was formally addressed by \citet{Maller1992,Maller1996}, who showed that the cure fraction is identifiable only when the support of the censoring distribution covers the support of the susceptible failure-time distribution. These results provide the theoretical basis for support-based assessments of follow-up sufficiency. However, the associated testing procedures do not determine the additional follow-up required when the available observation period is insufficient.

% To introduce the proposed framework, we first review the standard mixture cure model and the formal definition of follow-up sufficiency.

The standard mixture cure model introduced by \cite{Boag1949} and \cite{Berkson1952} is defined through an improper survival function \(S_p\) and assumes that the population comprises cured and susceptible individuals. Specifically, \(S_p\) is given by

\begin{equation}
S_p(t)=p+(1-p)S_0(t),
\label{mixtura}
\end{equation}

where \(p\) denotes the cure fraction and \(S_0(t)\) is the parametric survival function of susceptible individuals. As \(t\rightarrow\infty\), \(S_0(t)\rightarrow0\) and, consequently, \(S_p(t)\rightarrow p\), the cure proportion in the population.

% Since $S_0(t)\rightarrow0$ as $t\rightarrow\infty$, the population survival function converges to the cure fraction,
% that is, $S_p(t)\rightarrow p$. 

In Equation~\eqref{mixtura}, reliable estimation of the cure fraction requires follow-up sufficiently long for the susceptible survival function, \(S_0(t)\), to be observed close to zero, allowing the population survival curve to approach its limiting value \(p\).

Let
$$
\tau_{F_0}=\sup\{t:S_0(t)>0\}, \qquad \tau_C=\sup\{t:G(t)<1\},
$$

denote the upper endpoints of the failure-time distribution for susceptible individuals and the censoring distribution, respectively, where $G(t)$ is the cumulative distribution function of the censoring time.

% In this context, $\tau_{F_0}$ denotes the latest potential failure time among susceptible individuals, whereas $\tau_C$ denotes the largest observable follow-up time. Therefore, follow-up is sufficient only when the support of the censoring distribution covers that of the susceptible failure-time distribution.

In this context, \(\tau_{F_0}\) represents the latest time at which a susceptible individual may experience the event of interest, whereas \(\tau_C\) represents the longest follow-up attainable under the censoring mechanism. Thus, \(\tau_C\) determines the maximum observation window available for detecting failures.

Follow-up is considered sufficient if the support of the censoring distribution fully covers that of the susceptible failure-time distribution, that is,

% Formally, sufficient follow-up is defined by

\begin{equation}
\tau_C \geq \tau_{F_0}.
\label{eq:sufficient}
\end{equation}

Under this condition, every susceptible failure has the potential to be observed before the study terminates.

Conversely, if 

\begin{equation}
\tau_C < \tau_{F_0},
\label{eq:sufficient_contray}
\end{equation}

there is a nonempty interval \((\tau_C,\tau_{F_0}]\) over which susceptible failures cannot be observed. Consequently, individuals surviving beyond \(\tau_C\) cannot be distinguished from cured individuals, compromising identification of the cure fraction in an observed sample.

% Otherwise, some susceptible individuals remain unobserved after the end of follow-up, making it impossible to distinguish late failures from cured individuals.

The problem of follow-up sufficiency can therefore be formulated as the following hypothesis test:

\begin{equation*}
H_0:\tau_C \geq \tau_{F_0} \qquad \text{versus} \qquad H_1:\tau_C < \tau_{F_0}.
\end{equation*}

The null hypothesis states that the available follow-up is sufficiently long for reliable estimation of the cure fraction, whereas the alternative indicates that the study ends before all susceptible failures can potentially be observed. Several nonparametric procedures have been developed to test hypotheses based on the relative upper endpoints of the susceptible failure-time and censoring distributions \citep{Maller1992,Maller1996}. More recently, \citet{xie2024testing} proposed a bootstrap procedure with improved finite-sample performance.

A different diagnostic perspective was introduced by \citet{selukar2023receus}. Rather than relying solely on sample extremes, RECeUS uses a parametric estimate of the susceptible survival probability remaining at the administrative censoring time, standardized by the corresponding population survival probability. RECeUS assesses whether a cure model is appropriate at the observed end of follow-up, but it is not a formal hypothesis test and does not invert the fitted survival function to estimate the required minimum follow-up time.

It is therefore useful to distinguish three related problems: formal testing of follow-up sufficiency, diagnosis of cure-model appropriateness at the observed study endpoint, and estimation of a practical time-to-cure. The framework developed in the next section connects these perspectives without treating them as identical inferential targets.

\section{A parametric framework for estimating the minimum sufficient follow-up time} \label{sec:criteria}

In this section, we present a framework comprising three complementary components: the Parametric Follow-up Sufficiency Test (PFST), the Plateau Distance Criterion (PDC), and the Residual Survival Criterion (RSC). The PFST compares the terminal Kaplan-Meier estimate with the cure fraction estimated from a parametric mixture cure model. Inference for the PFST uses both a centered nonparametric bootstrap procedure and an asymptotic Gaussian approximation based on influence functions. The PDC and RSC are then used to estimate the minimum sufficient follow-up time from complementary population- and susceptible-level perspectives.

\subsection{The parametric follow-up sufficiency test}

In survival studies, the Kaplan-Meier estimator provides a nonparametric description of survival over time. An apparent plateau in the estimated survival curve is often interpreted as evidence of long-term survivors, suggesting that a proportion of individuals may be cured or immune to the event of interest. However, such a plateau does not itself guarantee that the available follow-up is sufficient, because susceptible individuals may experience additional failures after the study ends.

%{\color{red} Aqui seria bom colocar exemplos de curvas de sobrevida com individuos plateau para justificar o problema, uma com uma cauda longa e outra com uma cauda mais curta. Aí, no paragrafo acima, iriamos justificar.}

As discussed by \citet{maller2024mixture}, the behavior of the right tail and the relative extremes of the failure-time and censoring distributions play a central role in assessing follow-up sufficiency. Similarly, \citet{yuen2026testing} emphasize that visual inspection of a prolonged Kaplan-Meier plateau may be ambiguous and inadequate, potentially leading to overestimation of the cure fraction. This limitation can be illustrated by comparing susceptible survival distributions with shorter and longer right tails. With a shorter tail, most failures occur before the study ends, and the observed plateau is more likely to represent the cure fraction. With a longer tail, however, late failures may remain unobserved, allowing an apparent plateau to emerge before the population survival curve reaches its true limiting value.

Let \(\widehat{p}\) denote the maximum likelihood estimate of the cure fraction obtained from the fitted parametric mixture cure model in Equation~\eqref{mixtura}. Let \(\widehat{p}_{KM}\) denote the terminal Kaplan-Meier estimate, defined by
\begin{equation}
\widehat{p}_{KM}=\widehat{S}_{KM}(t_{\max}),
\label{eq:pKM_terminal}
\end{equation}
where
\[
t_{\max}=\max\{Y_1,\ldots,Y_n\}
\]
is the largest observed follow-up time. Under sufficient follow-up, \(\widehat{p}_{KM}\) provides a nonparametric estimator of the cure fraction.

When follow-up is sufficient, the Kaplan-Meier estimator is expected to have reached its long-term plateau. Consequently, both \(\widehat{p}_{KM}\) and \(\widehat{p}\) target the same population quantity, namely the cure fraction \(p\), and their difference should be close to zero. Conversely, when follow-up is insufficient, some susceptible individuals may remain event-free at the end of the study and thus be indistinguishable from cured individuals. In this case, the terminal Kaplan-Meier estimate may remain above the asymptotic cure plateau predicted by the fitted parametric model.

Motivated by this discrepancy, we define the PFST statistic as
\begin{equation}
T_n^{Par}=\widehat{p}_{KM}-\widehat{p}.
\label{eq:PFST_statistic}
\end{equation}

The PFST statistic measures the discrepancy between the terminal nonparametric survival estimate and the cure fraction predicted by the fitted parametric cure model. Values close to zero indicate agreement between the two estimators and are consistent with sufficient follow-up. Positive values indicate that the terminal Kaplan-Meier estimate remains above the estimated cure plateau and are therefore consistent with insufficient follow-up. Larger positive values provide stronger evidence that additional follow-up may be required.

Under a correctly specified cure model, negative values are not expected asymptotically under insufficient follow-up. Small negative values may nevertheless occur in finite samples because of sampling variability or model misspecification.

To investigate the large-sample properties of the PFST statistic, we consider its asymptotic distribution. Because the terminal Kaplan-Meier estimator and the maximum likelihood estimator of the cure fraction are calculated from the same censored sample, their dependence must be incorporated into the asymptotic variance. The following proposition establishes the asymptotic distribution of the PFST statistic under the stated assumptions and provides the theoretical basis for the bootstrap and influence-function implementations developed below.

\begin{proposition}[Asymptotic distribution of the PFST statistic]
\label{propo1}
Consider the PFST statistic \(T_n^{Par}=\widehat{p}_{KM}-\widehat{p}\). Suppose that the following regularity conditions hold:
\begin{enumerate}
    \item \textbf{Model specification:} the parametric mixture cure model is correctly specified;
    \item \textbf{Sufficient follow-up:} the terminal Kaplan-Meier estimator is consistent and asymptotically normal;
    \item \textbf{Joint asymptotic normality:} the estimators \(\widehat{p}_{KM}\) and \(\widehat{p}\) are consistent and jointly asymptotically normal, such that
    \begin{equation}
    \sqrt{n}
    \begin{pmatrix}
    \widehat{p}_{KM}-p\\
    \widehat{p}-p
    \end{pmatrix}
    \xrightarrow{d}
    \mathcal{N}
    \left[
    \begin{pmatrix}
    0\\
    0
    \end{pmatrix},
    \begin{pmatrix}
    \sigma_{KM}^{2} & \sigma_{12}\\
    \sigma_{12} & \sigma_{p}^{2}
    \end{pmatrix}
    \right],
    \label{eq:joint_asymptotic_distribution}
    \end{equation}
    where \(p\) is the true cure fraction and
    \[
    \sigma_{12}=\lim_{n\rightarrow\infty}n\,\operatorname{Cov}\left(\widehat{p}_{KM},\widehat{p}\right).
    \]
\end{enumerate}
Then,
\begin{equation}
\sqrt{n}\,T_n^{Par}\xrightarrow{d}\mathcal{N}(0,\tau^{2}),
\label{eq:PFST_asymptotic_distribution}
\end{equation}
where
\begin{equation}
\tau^{2}=\sigma_{KM}^{2}+\sigma_{p}^{2}-2\sigma_{12}.
\label{eq:tau_variance}
\end{equation}
Consequently, under sufficient follow-up,
\[
T_n^{Par}\xrightarrow{P}0
\qquad\text{as}\qquad
n\rightarrow\infty.
\]
\end{proposition}

\begin{proof}
Since \(T_n^{Par}=\widehat{p}_{KM}-\widehat{p}\), the PFST statistic is a linear combination of two jointly asymptotically normal estimators. Applying the continuous mapping theorem to the joint distribution in Equation~\eqref{eq:joint_asymptotic_distribution} gives
\[
\sqrt{n}\,T_n^{Par}
=
\begin{pmatrix}
1 & -1
\end{pmatrix}
\sqrt{n}
\begin{pmatrix}
\widehat{p}_{KM}-p\\
\widehat{p}-p
\end{pmatrix}
\xrightarrow{d}
\mathcal{N}\left(0,\sigma_{KM}^{2}+\sigma_{p}^{2}-2\sigma_{12}\right).
\]
Therefore,
\[
\sqrt{n}\,T_n^{Par}\xrightarrow{d}\mathcal{N}(0,\tau^{2}),
\]
where \(\tau^{2}=\sigma_{KM}^{2}+\sigma_{p}^{2}-2\sigma_{12}\). Moreover, because \(\sqrt{n}\,T_n^{Par}=O_p(1)\), it follows that
\[
T_n^{Par}=O_p\left(n^{-1/2}\right)\xrightarrow{P}0.
\]
\end{proof}

Although Proposition~\ref{propo1} establishes the asymptotic distribution of the PFST statistic, its variance depends on the unknown components \(\sigma_{KM}^{2}\), \(\sigma_{p}^{2}\), and \(\sigma_{12}\). Therefore, these quantities must be estimated before the asymptotic Gaussian approximation can be used in practice.

Two complementary inferential implementations are considered in this paper. The first is a centered nonparametric bootstrap procedure, which is adopted as the primary finite-sample inferential method. The second is an influence-function estimator of the asymptotic variance, developed in the next subsection, which provides an analytic implementation of Proposition~\ref{propo1} and allows the accuracy of the asymptotic Gaussian approximation to be evaluated.

For the centered nonparametric bootstrap, samples are generated by resampling the observed pairs \((Y_i,\Delta_i)\) with replacement. For each bootstrap sample, the parametric mixture cure model is refitted, the cure fraction is re-estimated, and the PFST statistic is recomputed. The empirical distribution of the centered bootstrap statistics is then used to approximate the sampling distribution of the PFST statistic and obtain critical values and \(p\)-values. The PDC and RSC subsequently address estimation of the minimum sufficient follow-up time.

\subsection{Influence-function representation and asymptotic variance estimation}
\label{subsec:influence_function_pfst}

Proposition~\ref{propo1} establishes the asymptotic distribution of the PFST statistic under joint asymptotic normality of the terminal Kaplan-Meier estimator and the maximum likelihood estimator of the cure fraction. In this subsection, we provide asymptotic linear representations for these estimators and obtain an empirical estimator of the asymptotic variance \(\tau^2\) defined in Proposition~\ref{propo1}.

Let
\[
O_i=(Y_i,\Delta_i),\qquad i=1,\ldots,n,
\]
denote independent and identically distributed right-censored observations, where \(Y_i\) is the observed follow-up time and \(\Delta_i\) is the event indicator.

Under suitable regularity conditions, suppose that the terminal Kaplan-Meier estimator and the parametric estimator of the cure fraction admit the asymptotically linear representations
\begin{equation}
\sqrt{n}\left(\widehat{p}_{KM}-p\right)=\frac{1}{\sqrt{n}}\sum_{i=1}^{n}IF_{KM}(O_i)+o_p(1),
\label{eq:IF_KM_representation}
\end{equation}
and
\begin{equation}
\sqrt{n}\left(\widehat{p}-p\right)=\frac{1}{\sqrt{n}}\sum_{i=1}^{n}IF_p(O_i)+o_p(1),
\label{eq:IF_p_representation}
\end{equation}
where \(IF_{KM}(O_i)\) and \(IF_p(O_i)\) denote the individual influence contributions of the terminal Kaplan-Meier estimator and the parametric cure-fraction estimator, respectively. If these contributions have mean zero and finite second moments, the multivariate central limit theorem applied to their joint representation yields the joint asymptotic normality assumed in Proposition~\ref{propo1}.

The influence-function representation for estimators based on right-censored data, including the Kaplan-Meier product-limit estimator, was developed by \citet{reid1981influence}. Let \(u_1<\cdots<u_J\) denote the distinct observed event times. At each event time \(u_j\), let \(r_j\) denote the number of individuals at risk immediately before \(u_j\), and let \(d_j\) denote the number of events at \(u_j\). The terminal Kaplan-Meier estimate introduced in the previous subsection can be written as
\begin{equation}
\widehat{p}_{KM}=\prod_{j=1}^{J}\left(1-\frac{d_j}{r_j}\right).
\label{eq:pKM_product_limit}
\end{equation}

For each individual \(i\), define
\begin{equation}
E_i=\sum_{j=1}^{J}\frac{\mathbf{1}\{Y_i=u_j,\Delta_i=1\}}{r_j-d_j},
\label{eq:E_i_KM}
\end{equation}
and
\begin{equation}
G_i=\sum_{j=1}^{J}\mathbf{1}\{u_j\leq Y_i\}\frac{d_j}{r_j(r_j-d_j)}.
\label{eq:G_i_KM}
\end{equation}

The quantity \(E_i\) is nonzero only when individual \(i\) experiences the event at one of the observed event times, whereas \(G_i\) accumulates the Greenwood increments up to the observed follow-up time \(Y_i\). Provided that \(r_j>d_j\) for all observed event times, an empirical influence contribution for \(\widehat{p}_{KM}\) is given by
\begin{equation}
\widehat{IF}_{KM,i}=-n\widehat{p}_{KM}\left(E_i-G_i\right).
\label{eq:IF_KM_empirical}
\end{equation}

Because \(t_{\max}\) depends on the observed sample, the use of Equation~\eqref{eq:IF_KM_empirical} relies on the sufficient-follow-up and regularity conditions stated in Proposition~\ref{propo1}, under which the terminal Kaplan-Meier estimator is consistent, asymptotically normal, and admits the representation in Equation~\eqref{eq:IF_KM_representation}.

Let
\[
\overline{IF}_{KM}=\frac{1}{n}\sum_{i=1}^{n}\widehat{IF}_{KM,i},
\]
and define the centered influence contributions as
\begin{equation}
\widetilde{IF}_{KM,i}=\widehat{IF}_{KM,i}-\overline{IF}_{KM}.
\label{eq:IF_KM_centered}
\end{equation}

The asymptotic variance component associated with \(\sqrt{n}(\widehat{p}_{KM}-p)\) can then be estimated by
\begin{equation}
\widehat{\sigma}_{KM}^{\,2}=\frac{1}{n}\sum_{i=1}^{n}\widetilde{IF}_{KM,i}^{\,2}.
\label{eq:sigma_KM_IF}
\end{equation}

This expression is equivalent to \(n\) times Greenwood's variance estimator for the Kaplan-Meier estimator \citep{greenwood1926natural,breslow1974large,reid1981influence},
\begin{equation}
\widehat{\sigma}_{KM}^{\,2}=n\widehat{p}_{KM}^{\,2}\sum_{j=1}^{J}\frac{d_j}{r_j(r_j-d_j)}.
\label{eq:sigma_KM_greenwood}
\end{equation}

For the parametric component, we specialize the general mixture cure model in Equation~\eqref{mixtura} by assuming a Weibull distribution for the failure times of susceptible individuals. Let \(a>0\) and \(b>0\) denote the Weibull shape and scale parameters, respectively, and let \(p\in(0,1)\) denote the cure fraction.

For numerical estimation, consider the unconstrained parameterization
\begin{equation}
\boldsymbol{\theta}
=
\begin{pmatrix}
\eta_a\\
\eta_b\\
\eta_p
\end{pmatrix}
=
\begin{pmatrix}
\log(a)\\
\log(b)\\
\operatorname{logit}(p)
\end{pmatrix},
\label{eq:theta_unconstrained}
\end{equation}
where
\[
a=\exp(\eta_a),\qquad b=\exp(\eta_b),\qquad p=g(\boldsymbol{\theta})=\frac{\exp(\eta_p)}{1+\exp(\eta_p)}.
\]

The survival and density functions of the susceptible failure-time distribution are, respectively,
\begin{equation}
S_0(t;a,b)=\exp\left\{-\left(\frac{t}{b}\right)^a\right\},
\label{eq:weibull_survival}
\end{equation}
and
\begin{equation}
f_0(t;a,b)=\frac{a}{b}\left(\frac{t}{b}\right)^{a-1}\exp\left\{-\left(\frac{t}{b}\right)^a\right\}.
\label{eq:weibull_density}
\end{equation}

Under independent and noninformative right censoring, and up to additive terms involving only the censoring distribution, the contribution of observation \(i\) to the log-likelihood for \(\boldsymbol{\theta}\) is
\begin{equation}
\ell_i(\boldsymbol{\theta})=\Delta_i\left[\log(1-p)+\log f_0(Y_i;a,b)\right]+(1-\Delta_i)\log\left[p+(1-p)S_0(Y_i;a,b)\right].
\label{eq:individual_loglikelihood}
\end{equation}

Let
\begin{equation}
\boldsymbol{s}_i(\boldsymbol{\theta})=\frac{\partial\ell_i(\boldsymbol{\theta})}{\partial\boldsymbol{\theta}}
\label{eq:individual_score}
\end{equation}
denote the individual score vector, and define the sensitivity matrix by
\begin{equation}
\boldsymbol{A}(\boldsymbol{\theta}_0)=-\mathbb{E}\left[\frac{\partial\boldsymbol{s}_i(\boldsymbol{\theta})}{\partial\boldsymbol{\theta}^{\mathsf T}}\bigg|_{\boldsymbol{\theta}=\boldsymbol{\theta}_0}\right],
\label{eq:sensitivity_matrix}
\end{equation}
where \(\boldsymbol{\theta}_0\) denotes the true parameter vector.

Under the general \(M\)-estimation framework described by \citet{stefanski2002calculus}, the maximum likelihood estimator admits the influence-function representation
\begin{equation}
IF_{\boldsymbol{\theta}}(O_i)=\boldsymbol{A}(\boldsymbol{\theta}_0)^{-1}\boldsymbol{s}_i(\boldsymbol{\theta}_0).
\label{eq:IF_theta}
\end{equation}

Since the cure fraction is a smooth function of \(\boldsymbol{\theta}\), the delta method gives
\begin{equation}
IF_p(O_i)=\dot g(\boldsymbol{\theta}_0)^{\mathsf T}\boldsymbol{A}(\boldsymbol{\theta}_0)^{-1}\boldsymbol{s}_i(\boldsymbol{\theta}_0),
\label{eq:IF_p}
\end{equation}
where
\begin{equation}
\dot g(\boldsymbol{\theta}_0)=\frac{\partial g(\boldsymbol{\theta})}{\partial\boldsymbol{\theta}}\bigg|_{\boldsymbol{\theta}=\boldsymbol{\theta}_0}
=
\begin{pmatrix}
0\\
0\\
p(1-p)
\end{pmatrix}.
\label{eq:gradient_p}
\end{equation}

For empirical estimation, let
\begin{equation}
\widehat{\boldsymbol{A}}=-\frac{1}{n}\sum_{i=1}^{n}\frac{\partial\boldsymbol{s}_i(\boldsymbol{\theta})}{\partial\boldsymbol{\theta}^{\mathsf T}}\bigg|_{\boldsymbol{\theta}=\widehat{\boldsymbol{\theta}}}.
\label{eq:Ahat_empirical}
\end{equation}

In maximum likelihood estimation, \(\widehat{\boldsymbol{A}}\) corresponds to the negative Hessian matrix of the total log-likelihood evaluated at \(\widehat{\boldsymbol{\theta}}\), divided by \(n\).

The estimated individual influence contribution of the cure-fraction estimator is
\begin{equation}
\widehat{IF}_{p,i}=\dot g(\widehat{\boldsymbol{\theta}})^{\mathsf T}\widehat{\boldsymbol{A}}^{-1}\boldsymbol{s}_i(\widehat{\boldsymbol{\theta}}).
\label{eq:IF_p_empirical}
\end{equation}

Let
\[
\overline{IF}_{p}=\frac{1}{n}\sum_{i=1}^{n}\widehat{IF}_{p,i},
\]
and define the centered influence contributions by
\begin{equation}
\widetilde{IF}_{p,i}=\widehat{IF}_{p,i}-\overline{IF}_{p}.
\label{eq:IF_p_centered}
\end{equation}

The asymptotic variance component associated with \(\sqrt{n}(\widehat{p}-p)\) is estimated by
\begin{equation}
\widehat{\sigma}_{p}^{\,2}=\frac{1}{n}\sum_{i=1}^{n}\widetilde{IF}_{p,i}^{\,2}.
\label{eq:sigma_p_IF}
\end{equation}

Because \(a\), \(b\), and \(p\) are estimated jointly, Equation~\eqref{eq:sigma_p_IF} incorporates both the direct uncertainty associated with the cure fraction and the uncertainty propagated through the estimation of the Weibull shape and scale parameters.

Combining Equations~\eqref{eq:IF_KM_representation} and~\eqref{eq:IF_p_representation}, the PFST statistic satisfies
\begin{equation}
\sqrt{n}T_n^{Par}=\frac{1}{\sqrt{n}}\sum_{i=1}^{n}IF_T(O_i)+o_p(1),
\label{eq:IF_T_representation}
\end{equation}
where
\begin{equation}
IF_T(O_i)=IF_{KM}(O_i)-IF_p(O_i).
\label{eq:IF_T}
\end{equation}

Since \(IF_T(O_i)=IF_{KM}(O_i)-IF_p(O_i)\), its variance satisfies
\begin{align}
\operatorname{Var}\left\{IF_T(O_i)\right\}
&=\operatorname{Var}\left\{IF_{KM}(O_i)\right\}+\operatorname{Var}\left\{IF_p(O_i)\right\}\nonumber\\
&\quad-2\operatorname{Cov}\left\{IF_{KM}(O_i),IF_p(O_i)\right\}.
\label{eq:IF_T_variance_expansion}
\end{align}

This variance expansion recovers the asymptotic variance \(\tau^2\) defined in Proposition~\ref{propo1}. The covariance component \(\sigma_{12}\), also defined in Proposition~\ref{propo1}, has the influence-function representation
\[
\sigma_{12}=\operatorname{Cov}\left\{IF_{KM}(O_i),IF_p(O_i)\right\}.
\]

Because \(\widehat{p}_{KM}\) and \(\widehat{p}\) are calculated from the same censored sample, their dependence must be incorporated into the variance of the PFST statistic. The covariance component can be estimated from the empirical cross-product of their centered individual influence contributions:
\begin{equation}
\widehat{\sigma}_{12}=\frac{1}{n}\sum_{i=1}^{n}\widetilde{IF}_{KM,i}\widetilde{IF}_{p,i}.
\label{eq:sigma_12_IF}
\end{equation}

The resulting influence-function estimator of the asymptotic variance of the PFST statistic is
\begin{equation}
\widehat{\tau}_{IF}^{\,2}=\widehat{\sigma}_{KM}^{\,2}+\widehat{\sigma}_{p}^{\,2}-2\widehat{\sigma}_{12}.
\label{eq:tau_IF_decomp}
\end{equation}

Equivalently, this variance can be obtained directly from the influence contributions of the difference:
\begin{equation}
\widehat{\tau}_{IF}^{\,2}=\frac{1}{n}\sum_{i=1}^{n}\left(\widetilde{IF}_{KM,i}-\widetilde{IF}_{p,i}\right)^2.
\label{eq:tau_IF_direct}
\end{equation}

Equations~\eqref{eq:tau_IF_decomp} and~\eqref{eq:tau_IF_direct} are algebraically equivalent. Let
\[
\widehat{\tau}_{IF}=\left(\widehat{\tau}_{IF}^{\,2}\right)^{1/2}
\]
denote the estimated asymptotic standard deviation. The standardized PFST statistic based on the influence-function representation is then defined as
\begin{equation}
Z_{IF}=\frac{\sqrt{n}T_n^{Par}}{\widehat{\tau}_{IF}}.
\label{eq:Z_IF}
\end{equation}

Under \(H_0\) and the regularity conditions stated in Proposition~\ref{propo1},
\begin{equation}
Z_{IF}\xrightarrow{d}\mathcal{N}(0,1).
\label{eq:Z_IF_asymptotic}
\end{equation}

For the one-sided alternative of insufficient follow-up, the corresponding \(p\)-value is
\begin{equation}
p_{IF}=1-\Phi(Z_{IF}),
\label{eq:pvalue_IF}
\end{equation}
and \(H_0\) is rejected at significance level \(\alpha\) whenever
\begin{equation}
Z_{IF}>z_{1-\alpha}.
\label{eq:decision_IF}
\end{equation}

Equivalently, the rejection rule can be expressed as
\begin{equation}
T_n^{Par}>z_{1-\alpha}\frac{\widehat{\tau}_{IF}}{\sqrt{n}}.
\label{eq:decision_IF_critical}
\end{equation}

The influence-function approach provides an analytic implementation of the asymptotic distribution established in Proposition~\ref{propo1}. The centered nonparametric bootstrap remains the primary finite-sample inferential procedure for the PFST, whereas the influence-function approximation is used to evaluate the finite-sample accuracy of the asymptotic Gaussian approximation.

\subsection{Plateau distance criterion (PDC)}

The PDC defines the minimum sufficient follow-up time according to the proximity of the population survival function to its asymptotic cure plateau. Under the standard mixture cure model, the cure fraction \(p\) is the limiting value of the population survival function. According to the PDC, follow-up is therefore considered sufficiently long once the remaining difference between the population survival probability and the cure fraction falls below a prespecified tolerance.

Let \(0<\delta<1-p\) denote a prespecified tolerance representing the maximum acceptable absolute deviation between the population survival function and the cure plateau. The minimum sufficient follow-up time is defined as
\begin{equation}
t_P(\delta)=\inf\left\{t\geq0:S_p(t)-p\leq\delta\right\}.
\label{eq:tP_definition}
\end{equation}

Using the mixture cure model in Equation~\eqref{mixtura}, the condition in Equation~\eqref{eq:tP_definition} is equivalent to \(S_0(t)\leq\delta/(1-p)\). Thus, for any parametric mixture cure model with an invertible susceptible survival function,
\begin{equation}
t_P(\delta)=S_0^{-1}\left(\frac{\delta}{1-p}\right).
\label{eq:general_inverse}
\end{equation}

Under the Weibull specification in Equation~\eqref{eq:weibull_survival}, substituting the susceptible survival function into Equation~\eqref{eq:general_inverse} yields
\begin{equation}
t_P(\delta)=b\left[-\log\left(\frac{\delta}{1-p}\right)\right]^{1/a}.
\label{eq:tP}
\end{equation}

In practice, replacing the unknown parameters \(p\), \(a\), and \(b\) by their maximum likelihood estimates gives
\begin{equation}
\widehat{t}_P(\delta)=\widehat{b}\left[-\log\left(\frac{\delta}{1-\widehat{p}}\right)\right]^{1/\widehat{a}}.
\label{eq:tP_hat}
\end{equation}

For fixed values of \((\widehat{p},\widehat{a},\widehat{b})\), Equation~\eqref{eq:tP_hat} shows that smaller values of \(\delta\) impose a more stringent definition of follow-up sufficiency and therefore produce larger estimates of the minimum required follow-up. The PDC provides a population-level measure by identifying the earliest time at which the fitted population survival function is within a prespecified tolerance of the estimated cure plateau.

\subsection{Residual survival criterion (RSC)}

The RSC defines the minimum sufficient follow-up time on the susceptible survival scale. Under this criterion, follow-up is considered sufficiently long once the probability that a susceptible individual remains event-free beyond a given time falls below a prespecified tolerance.

Let \(0<\varepsilon<1\) denote a prespecified tolerance representing the maximum acceptable residual survival probability among susceptible individuals. The minimum sufficient follow-up time is defined as
\begin{equation}
t_R(\varepsilon)=\inf\left\{t\geq0:S_0(t)\leq\varepsilon\right\}.
\label{eq:tR_definition}
\end{equation}

For a general parametric mixture cure model, whenever the inverse survival function of the susceptible population is available, Equation~\eqref{eq:tR_definition} can be written as
\begin{equation}
t_R(\varepsilon)=S_0^{-1}(\varepsilon).
\label{eq:tR_general}
\end{equation}

Under the Weibull specification in Equation~\eqref{eq:weibull_survival}, substituting the susceptible survival function into Equation~\eqref{eq:tR_general} yields
\begin{equation}
t_R(\varepsilon)=b\left[-\log(\varepsilon)\right]^{1/a}.
\label{eq:tR}
\end{equation}

The quantity \(t_R(\varepsilon)\) is the earliest follow-up time at which no more than a proportion \(\varepsilon\) of susceptible individuals is expected to remain event-free. Smaller values of \(\varepsilon\) correspond to more stringent follow-up requirements and therefore produce larger minimum follow-up times. In practice, replacing the unknown parameters with their maximum likelihood estimates yields
\begin{equation}
\widehat{t}_R(\varepsilon)=\widehat{b}\left[-\log(\varepsilon)\right]^{1/\widehat{a}}.
\label{eq:tR_hat}
\end{equation}

The RSC provides a susceptible-level measure of follow-up sufficiency by controlling the remaining probability of future failures among susceptible individuals.

\subsection{Relationship between the PDC and the RSC}
\label{subsec:relationship_between_criteria}

The two criteria defined in the previous subsections are motivated by different perspectives on follow-up sufficiency. The PDC is defined on the population survival scale and controls the distance between the population survival curve and the asymptotic cure plateau. In contrast, the RSC is formulated on the susceptible survival scale and controls the remaining probability of future failures among susceptible individuals.

Establishing the relationship between these criteria shows that the proposed framework does not introduce competing definitions of follow-up sufficiency. Instead, it provides two complementary formulations of the same underlying concept, allowing investigators to assess follow-up sufficiency from either a population- or susceptible-level perspective.

Under the standard mixture cure model in Equation~\eqref{mixtura}, the relationship between the two criteria follows from the identity

\begin{equation}
    S_p(t)-p=(1-p)S_0(t),
    \label{eq:decomposition}
\end{equation}

% the relationship between the two criteria follows immediately from the identity

Substituting Equation~\eqref{eq:decomposition} into the definition of the PDC gives

% {\color{blue} A Equação do modelo de mistura padrão está sendo repetida, será que não podemos apenas mencionar??}
% {\color{red}

% Under the standard mixture cure model, the relationship between the two criteria follows immediately from the identity

% \begin{equation}
% S_p(t)-p=(1-p)S_0(t),
% \label{eq:decomposition}
% \end{equation}

% which decomposes the population survival function into the cure fraction and the survival contribution of susceptible individuals.
% }
% Substituting Equation~\eqref{eq:decomposition} into the definition of the PDC gives

\[
(1-p)S_0(t)\le\delta,
\]

which is equivalent to

\[
S_0(t)\le\frac{\delta}{1-p}.
\]

Hence, the two criteria are equivalent whenever

\begin{equation}
\varepsilon=\frac{\delta}{1-p}.
\label{eq:relation}
\end{equation}

The criteria differ not in the resulting follow-up time but in the scientific interpretation of the tolerance parameter. Whereas \(\delta\) measures proximity to the cure plateau on the population survival scale, \(\varepsilon\) quantifies the remaining survival probability among susceptible individuals.

The PDC is naturally expressed on the population survival scale and is most appropriate when the objective is to quantify how closely the overall survival curve has approached the cure plateau or to assess the stabilization of the estimated cure fraction. In contrast, the RSC is formulated on the susceptible survival scale and directly controls the residual probability that a susceptible individual remains event-free beyond the end of follow-up, making it preferable when the remaining uncertainty among susceptible individuals is of primary interest.

Under the transformation in Equation~\eqref{eq:relation}, both criteria produce the same minimum sufficient follow-up time,

\[
t_P(\delta)=t_R(\varepsilon).
\]

Researchers may therefore select the criterion that best reflects the scientific objective of their study without affecting the resulting estimate of the minimum sufficient follow-up time.

% Although mathematically equivalent, the two criteria provide complementary perspectives on follow-up sufficiency. The PDC is appropriate when the objective is to quantify the proximity of the population survival curve to the cure plateau, whereas the RSC is preferable when the interest lies in controlling the remaining probability of future failures among susceptible individuals. Researchers may therefore select the criterion that best reflects the scientific objective of their study without affecting the resulting estimate of the minimum sufficient follow-up time.

\section{Simulation study}
\label{sec:simulation_study}
The simulation study has two objectives. First, we investigate the inferential performance of the proposed PFST in finite samples by comparing the centered bootstrap procedure with its asymptotic influence-function approximation. Second, we evaluate the finite-sample behavior of the proposed estimators of the minimum sufficient follow-up time under different censoring intensities, cure fractions, and sample sizes.

% The simulation study is divided into two complementary parts. The first evaluates the finite-sample performance of the PFST by comparing the centered nonparametric bootstrap with the asymptotic influence-function procedure. The second investigates the finite-sample behavior of the minimum sufficient follow-up times estimated by the PDC and the RSC.

In both parts, susceptible failure times are generated from a Weibull distribution with shape parameter \(a=1.5\) and scale parameter \(b=1.5\). The two experiments differ in their estimation strategy. In the PFST experiment, \(a\) and \(b\) are fixed at their true generating values and only the cure fraction \(p\) is estimated. This controlled design isolates the comparison between the bootstrap and influence-function procedures. In the PDC and RSC experiment, \(a\), \(b\), and \(p\) are estimated jointly by maximum likelihood in each Monte Carlo replication. This second design incorporates the sampling uncertainty associated with estimation of the susceptible survival distribution and therefore provides a more realistic assessment of the finite-sample behavior of the estimated minimum follow-up times.

Individuals are independently classified as cured with probability \(p\), in which case their failure time is set to infinity. Administrative censoring times are independently generated from a \(\operatorname{Uniform}(0,\lambda)\) distribution. For individual \(i\), the observed follow-up time and event indicator are
\[
Y_i=\min(T_i,C_i)
\qquad\text{and}\qquad
\Delta_i=\mathbf{1}\{T_i\leq C_i\},
\]
respectively, where \(T_i\) denotes the failure time and \(C_i\) denotes the administrative censoring time.

The simulation considers
\[
N=\{100,500,1000,5000,10000\},\qquad
P=\{0.25,0.50,0.75\},\qquad
\Lambda=\{2,3.5,5\}.
\]

The three values of $\lambda$ represent short, intermediate, and long administrative follow-up periods, respectively, thereby generating increasing degrees of follow-up adequacy.

% The simulation considers
% \[
% \mathcal{N}=\{100,500,1000,5000,10000\},\qquad
% \mathcal{P}=\{0.25,0.50,0.75\},\qquad
% \Lambda=\{2,3.5,5\}.
% \]
% The PDC is evaluated at
% \[
% \mathcal{D}=\{0.050,0.025,0.010\},
% \]
% whereas the RSC is evaluated at
% \[
% \mathcal{E}=\{0.050,0.025,0.010\}.
% \]

The PDC is evaluated at

\[
D=\{0.050,0.025,0.010\},
\]

whereas the RSC is evaluated at

\[
E=\{0.050,0.025,0.010\}.
\]

These tolerance values represent increasingly stringent definitions of follow-up sufficiency. Smaller tolerances require the population survival curve (or, equivalently, the susceptible survival probability) to approach its limiting value more closely before follow-up is considered sufficient.

Apart from the different estimation strategies, both simulation experiments use the same data-generating mechanism, parameter grid, and tolerance levels, but they were conducted as separate Monte Carlo experiments because they employ different estimation strategies and numbers of replications. The PFST experiment uses \(R=500\) replications with \(a=b=1.5\) fixed, whereas the PDC and RSC experiment uses \(R=1000\) replications with \(a\), \(b\), and \(p\) estimated jointly. Within the second experiment, the PDC and RSC are calculated from the same fitted model in each replication, allowing their finite-sample behavior to be compared under identical simulated samples.

Because the susceptible Weibull distribution has unbounded upper support, whereas the administrative censoring distribution has upper endpoint \(\lambda<\infty\), all simulated scenarios satisfy \(H_1:\tau_C<\tau_{F_0}\).

Figure~\ref{fig:plateau_lambda235} illustrates the common data-generating mechanism under the three administrative censoring scenarios for \(p=0.50\). The fitted mixture cure curve displayed in the figure corresponds to the controlled PFST experiment, in which \(a=b=1.5\) are fixed and only \(p\) is estimated.

\begin{figure}[!ht]
\centering
\begin{subfigure}[t]{0.32\textwidth}
\centering
\includegraphics[width=\linewidth]{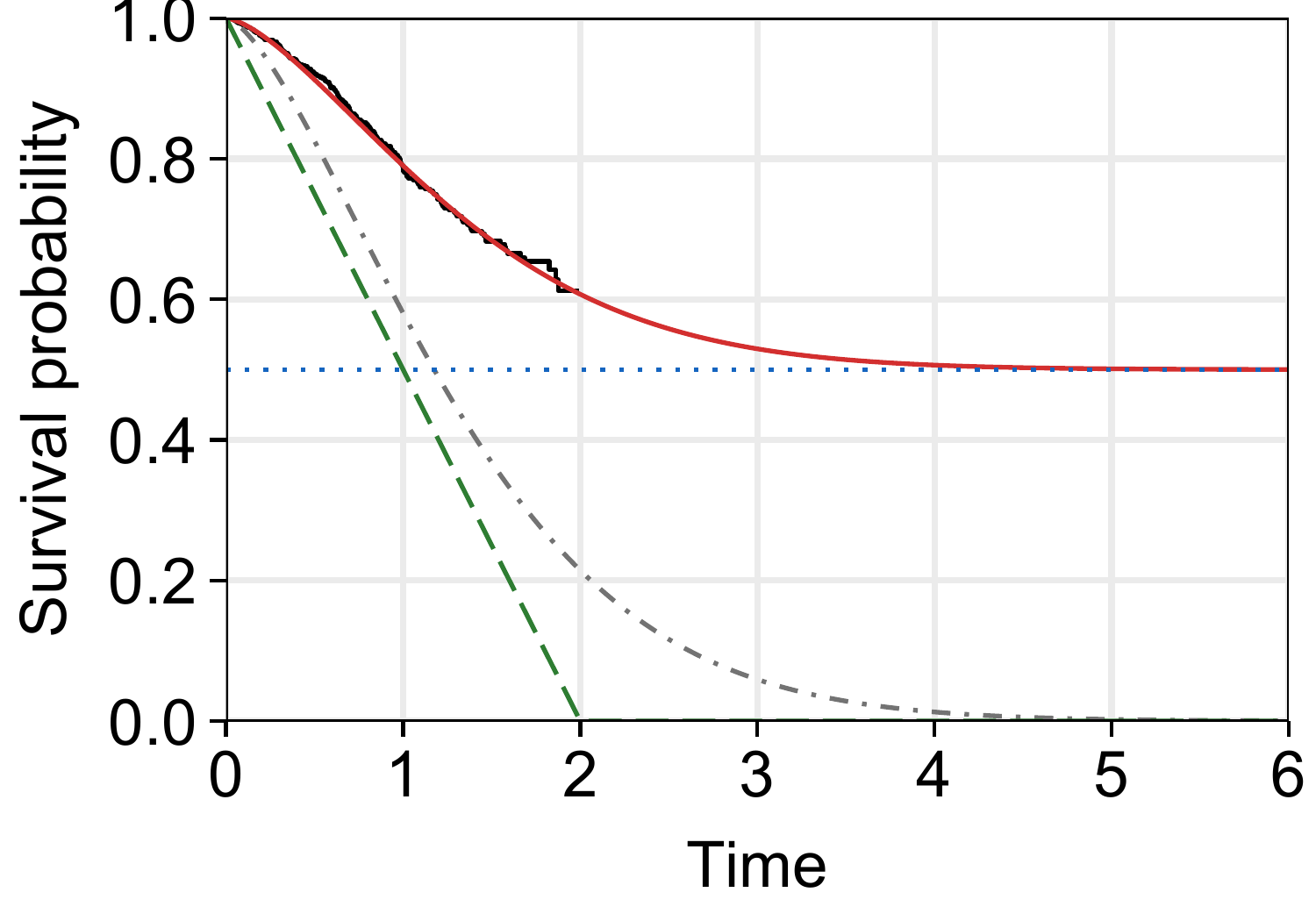}
\caption{\(\lambda=2\)}
\label{fig:plateau_lambda2}
\end{subfigure}
\hfill
\begin{subfigure}[t]{0.32\textwidth}
\centering
\includegraphics[width=\linewidth]{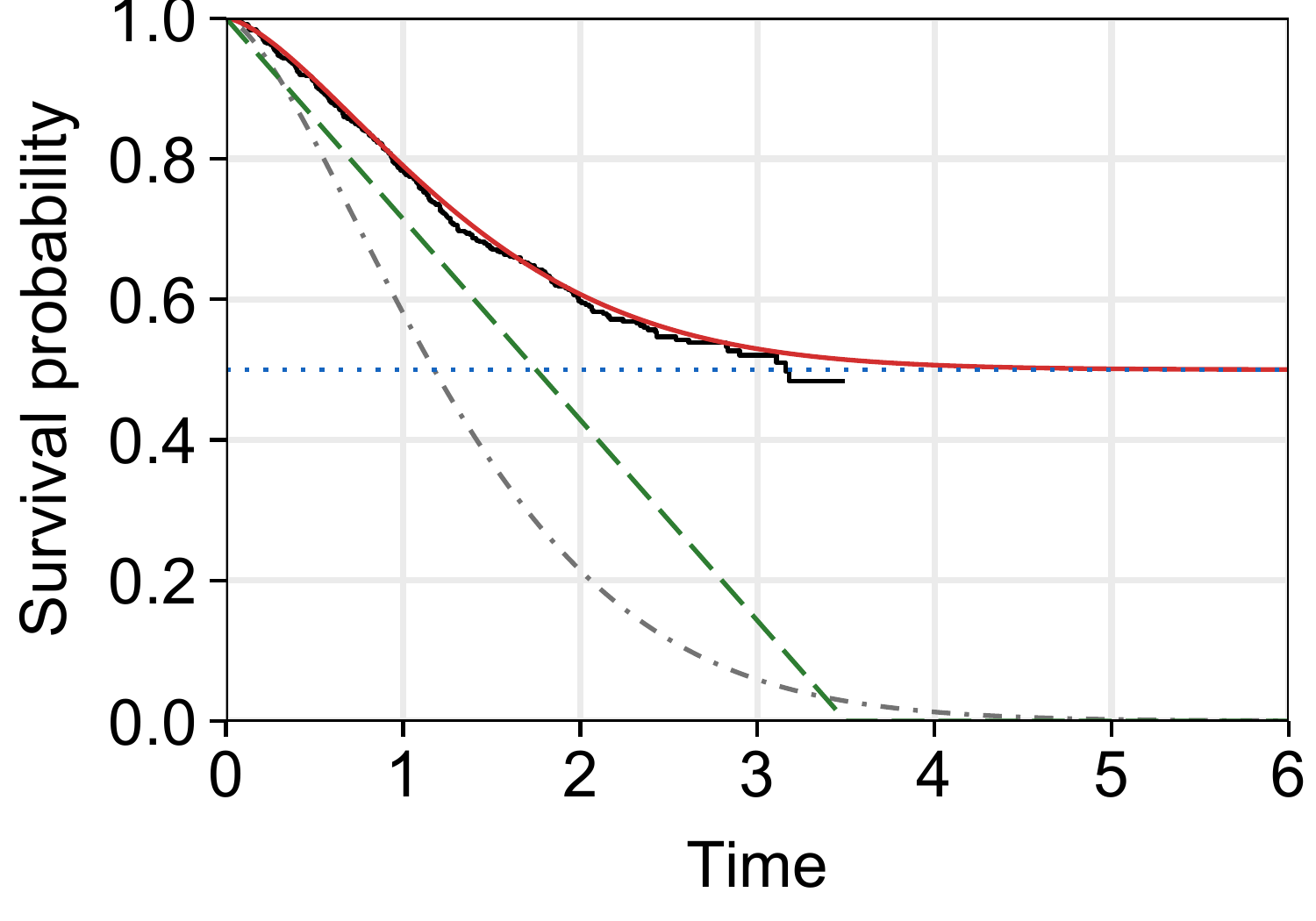}
\caption{\(\lambda=3.5\)}
\label{fig:plateau_lambda35}
\end{subfigure}
\hfill
\begin{subfigure}[t]{0.32\textwidth}
\centering
\includegraphics[width=\linewidth]{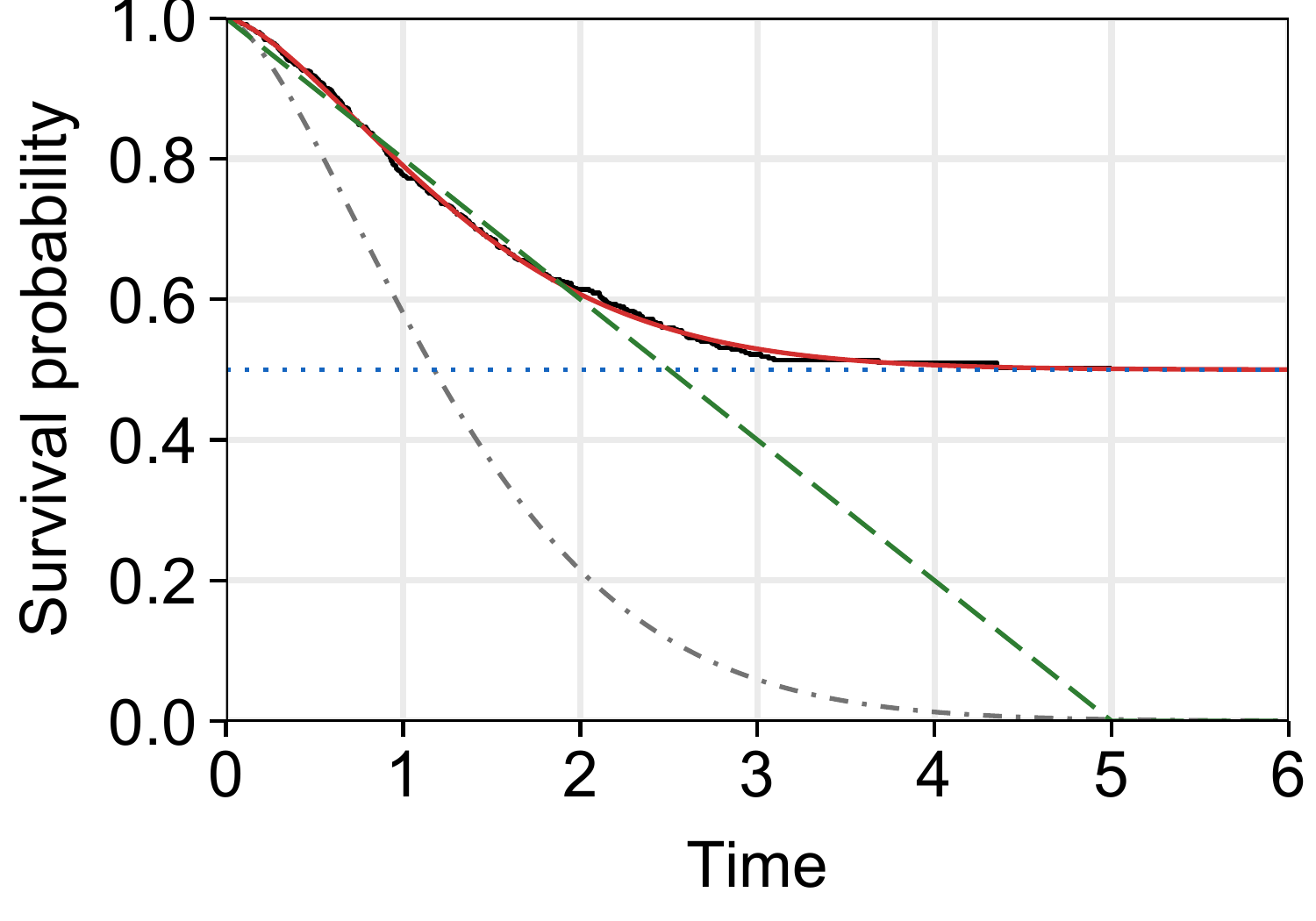}
\caption{\(\lambda=5\)}
\label{fig:plateau_lambda5}
\end{subfigure}
\caption{Illustration of the common data-generating mechanism under different administrative censoring distributions. Susceptible failure times follow a Weibull distribution with \(a=b=1.5\), the cure fraction is \(p=0.50\), and censoring times follow a \(\operatorname{Uniform}(0,\lambda)\) distribution. The black step function is the Kaplan-Meier estimator, the red curve is the fitted Weibull mixture cure model from the controlled PFST experiment, in which only \(p\) is estimated, the gray dot-dashed curve is the theoretical Weibull survival function for susceptible individuals, and the green long-dashed curve is the theoretical survival function of the administrative censoring time. The blue horizontal dotted line indicates the estimated cure fraction.}
\label{fig:plateau_lambda235}
\end{figure}

\subsection{Finite-sample performance of the PFST}
\label{subsec:simulation_pfst}

We evaluate the PFST using the centered nonparametric bootstrap and the asymptotic influence-function procedure described in Section~\ref{sec:criteria}. Both procedures are applied to the same simulated datasets and use the statistic \(\Tparam\) defined in Equation~\eqref{eq:PFST_statistic}. For each combination of \(n\), \(p\), and \(\lambda\), we perform \(R=500\) Monte Carlo replications. The nominal significance level is \(\alpha=0.05\), and the bootstrap procedure uses \(B=500\) valid resamples.

In each original sample, \(p\) is estimated by maximum likelihood with \(a=b=1.5\) fixed. The terminal Kaplan-Meier estimate \(\pKM\) is then calculated and combined with \(\phat\) to obtain \(\Tparam\). In each bootstrap resample, the observed pairs \((Y_i,\Delta_i)\) are sampled with replacement, \(a\) and \(b\) remain fixed at \(1.5\), and only \(p\) is re-estimated. The asymptotic procedure is the one-dimensional specialization of the influence-function representation developed in Subsection~\ref{subsec:influence_function_pfst}, with \(\eta_p=\operatorname{logit}(p)\) as the only unknown parametric component.

To ensure a paired comparison, we calculate the empirical rejection percentages of the bootstrap and asymptotic procedures from the same jointly valid Monte Carlo replications. A replication is jointly valid if the estimate of \(p\) is interior and numerically valid, the terminal Kaplan-Meier estimator can be calculated, the influence-function variance is finite and positive, and \(B\) valid bootstrap statistics are obtained. Algorithm~\ref{alg:simulation_pfst} summarizes the Monte Carlo procedure used to evaluate the PFST.

\begin{algorithm}[ht!]
\footnotesize
\caption{Monte Carlo simulation procedure for the PFST}
\label{alg:simulation_pfst}
\begin{algorithmic}[1]
\Require Sample sizes \(\mathcal{N}\), cure fractions \(\mathcal{P}\),
censoring limits \(\Lambda\), fixed Weibull parameters \(a=b=1.5\),
\(R=500\), \(B=500\), \(\alpha=0.05\), and at most \(5B\) bootstrap attempts
\For{each \(\lambda\in\Lambda\)}
    \For{each \(p\in\mathcal{P}\)}
        \For{each \(n\in\mathcal{N}\)}
            \For{\(r=1,\ldots,R\)}
                \State Generate cure indicators \(Z_i\sim\operatorname{Bernoulli}(p)\)
                \State Generate \(T_i\sim\operatorname{Weibull}(1.5,1.5)\) if
                \(Z_i=0\), set \(T_i=\infty\) if \(Z_i=1\), and generate
                \(C_i\sim\operatorname{Uniform}(0,\lambda)\)
                \State Set \(Y_i=\min(T_i,C_i)\) and
                \(\Delta_i=\mathbf{1}\{T_i\leq C_i\}\)
                \State Estimate \(p\) by maximum likelihood with \(a=b=1.5\)
                fixed, obtaining \(\widehat p\)
                \If{the fit is numerically valid and \(\widehat p\in(0,1)\)}
                    \State Calculate the terminal Kaplan-Meier estimate
                    \(\widehat p_{KM}\) and
                    \(T_n^{Par}=\widehat p_{KM}-\widehat p\)
                    \State Calculate the centered joint influence contributions
                    \(\widetilde{IF}_{T,i}
                    =\widetilde{IF}_{KM,i}-\widetilde{IF}_{p,i}\)
                    \State Set
                    \(\widehat{\tau}_{IF}^{\,2}
                    =n^{-1}\sum_{i=1}^{n}\widetilde{IF}_{T,i}^{\,2}\)
                    and
                    \(c_{IF}=z_{1-\alpha}\widehat{\tau}_{IF}/\sqrt n\)
                    \State Record
                    \(I_{IF}=\mathbf{1}\{T_n^{Par}>c_{IF}\}\)
                    \State Initialize \(b=0\), \(m=0\), and
                    \(\mathcal{D}^{*}=\varnothing\)
                    \While{\(b<B\) and \(m<5B\)}
                        \State Set \(m\gets m+1\) and resample the observed pairs
                        \((Y_i,\Delta_i)\) with replacement
                        \State Re-estimate only \(p\), calculate
                        \(T_{n,m}^{Par,*}\), and verify numerical validity
                        \If{the bootstrap fit and terminal Kaplan-Meier estimate are valid}
                            \State Set \(b\gets b+1\) and append
                            \(D_b^{*}=T_{n,m}^{Par,*}-T_n^{Par}\) to
                            \(\mathcal{D}^{*}\)
                        \EndIf
                    \EndWhile
                    \If{\(b=B\) and \(\widehat{\tau}_{IF}^{\,2}\) is finite and positive}
                        \State Set
                        \(c_{\mathrm{boot}}=Q_{1-\alpha}(\mathcal{D}^{*})\)
                        and record
                        \(I_{\mathrm{boot}}
                        =\mathbf{1}\{T_n^{Par}>c_{\mathrm{boot}}\}\)
                        \State Mark the Monte Carlo replication as jointly valid
                    \Else
                        \State Mark the Monte Carlo replication as invalid
                    \EndIf
                \Else
                    \State Mark the Monte Carlo replication as invalid
                \EndIf
            \EndFor
            \State Let \(R_v\) be the number of jointly valid replications
            \State Compute
            \(100R_v^{-1}\sum I_{\mathrm{boot}}\) and
            \(100R_v^{-1}\sum I_{IF}\)
        \EndFor
    \EndFor
\EndFor
\end{algorithmic}
\end{algorithm}

Since all scenarios satisfy \(H_1\), the rejection percentages in Table~\ref{tab:pfst_bootstrap_asymptotic} represent empirical power.

\begin{table}[ht!]
\centering
\scriptsize
\caption{Empirical power of the PFST obtained using the centered nonparametric bootstrap and the asymptotic influence-function procedure, with \(a=b=1.5\) fixed at their true generating values and only \(p\) estimated. Results are based on \(R=500\), \(B=500\), and \(\alpha=0.05\).}
\label{tab:pfst_bootstrap_asymptotic}
\setlength{\tabcolsep}{7pt}
\renewcommand{\arraystretch}{0.90}
\begin{adjustbox}{max width=\textwidth,max totalheight=0.88\textheight,keepaspectratio}
\begin{tabular}{cccccccc}
\toprule
\multirow[c]{2}{*}{\(p\)}
& \multirow[c]{2}{*}{\(n\)}
& \multicolumn{2}{c}{\(\lambda=2.0\)}
& \multicolumn{2}{c}{\(\lambda=3.5\)}
& \multicolumn{2}{c}{\(\lambda=5.0\)} \\
\cmidrule(lr){3-4}
\cmidrule(lr){5-6}
\cmidrule(lr){7-8}
& &
\makecell{Bootstrap rejection\\of \(H_0\) (\%)}
& \makecell{Asymptotic IF rejection\\of \(H_0\) (\%)}
& \makecell{Bootstrap rejection\\of \(H_0\) (\%)}
& \makecell{Asymptotic IF rejection\\of \(H_0\) (\%)}
& \makecell{Bootstrap rejection\\of \(H_0\) (\%)}
& \makecell{Asymptotic IF rejection\\of \(H_0\) (\%)} \\
\midrule
\multirow[c]{5}{*}{\(0.25\)}
& 100   & 77.39 & 76.99 & 40.12 & 41.53 & 17.51 & 19.11 \\
& 500   & 91.55 & 90.74 & 49.70 & 48.70 & 19.00 & 17.80 \\
& 1000  & 95.78 & 94.98 & 46.89 & 45.29 & 20.40 & 19.00 \\
& 5000  & 99.60 & 99.40 & 78.76 & 76.55 & 24.20 & 21.60 \\
& 10000 & 100.00 & 99.80 & 85.20 & 83.80 & 22.60 & 21.00 \\
\cmidrule(lr){1-8}
\multirow[c]{5}{*}{\(0.50\)}
& 100   & 71.49 & 71.29 & 36.60 & 35.80 & 17.00 & 17.20 \\
& 500   & 89.60 & 88.20 & 42.20 & 40.80 & 18.44 & 17.03 \\
& 1000  & 94.20 & 93.60 & 47.20 & 43.60 & 16.40 & 14.80 \\
& 5000  & 98.80 & 98.20 & 68.60 & 67.60 & 16.20 & 15.00 \\
& 10000 & 99.20 & 99.00 & 82.40 & 80.20 & 22.20 & 20.00 \\
\cmidrule(lr){1-8}
\multirow[c]{5}{*}{\(0.75\)}
& 100   & 64.13 & 64.93 & 32.00 & 32.60 & 17.60 & 17.40 \\
& 500   & 80.00 & 77.20 & 38.20 & 35.40 & 15.20 & 15.20 \\
& 1000  & 88.40 & 87.20 & 38.60 & 36.20 & 13.20 & 12.20 \\
& 5000  & 98.00 & 97.20 & 51.80 & 49.60 & 15.40 & 14.40 \\
& 10000 & 98.60 & 97.60 & 64.20 & 63.20 & 14.80 & 14.20 \\
\bottomrule
\end{tabular}
\end{adjustbox}
\end{table}

The bootstrap and asymptotic influence-function procedures yield similar empirical power throughout the simulation grid. The largest absolute difference between their rejection percentages is \(3.6\) percentage points, observed for \(\lambda=3.5\), \(p=0.50\), and \(n=1000\); differences are smaller in the remaining scenarios. This close agreement indicates that the asymptotic influence-function approximation reproduces the centered bootstrap well when the susceptible failure-time distribution is correctly specified, its parameters are treated as known, and only the cure fraction is estimated.

For \(\lambda=2\), empirical power is already moderately high at \(n=100\) and increases rapidly with sample size, approaching \(100\%\) in all cure-fraction scenarios. For \(\lambda=3.5\), power is intermediate and increases with \(n\), with generally lower rejection percentages as the cure fraction increases. For \(\lambda=5\), empirical power remains low even in large samples because the longer observation period reduces the discrepancy between \(\pKM\) and \(\phat\), making it more difficult to detect in finite samples.

The agreement observed in this controlled experiment does not imply that the susceptible-distribution parameters must be known when applying the PFST. Rather, it provides evidence that the asymptotic variance construction is coherent when the susceptible component is correctly specified and the cure fraction is the primary unknown parametric component.

\subsection{Finite-sample behavior of the minimum sufficient follow-up criteria}
\label{subsec:simulation_criteria}

The second part of the simulation evaluates the PDC and RSC using \(R=1000\) Monte Carlo replications. In each replication, the Weibull shape parameter \(a\), scale parameter \(b\), and cure fraction \(p\) are estimated jointly by maximum likelihood. For every numerically valid fitted model, the corresponding minimum sufficient follow-up times are calculated from Equations~\eqref{eq:tP_hat} and~\eqref{eq:tR_hat}. The results are summarized by their Monte Carlo means and sample standard deviations.

Because the PDC depends on \((\widehat a,\widehat b,\widehat p)\), its sampling variability reflects uncertainty in all three model parameters. The RSC depends on \((\widehat a,\widehat b)\) and therefore also varies among replications, even though its theoretical target does not depend on the cure fraction. The cure fraction can nevertheless affect the finite-sample precision of the RSC estimator indirectly because larger cure fractions reduce the expected number of susceptible individuals and observed failures.

The Monte Carlo procedure used to evaluate the two criteria is summarized in Algorithm~\ref{alg:simulation_criteria}.

\begin{algorithm}[ht!]
\footnotesize
\caption{Monte Carlo simulation procedure for the PDC and RSC}
\label{alg:simulation_criteria}
\begin{algorithmic}[1]
\Require Sample sizes \(\mathcal{N}\), cure fractions \(\mathcal{P}\),
censoring limits \(\Lambda\), tolerances
\(\mathcal{D}=\mathcal{E}=\{0.050,0.025,0.010\}\), true Weibull
parameters \(a=b=1.5\), and \(R=1000\)
\For{each \(\lambda\in\Lambda\)}
    \For{each \(p\in\mathcal{P}\)}
        \For{each \(n\in\mathcal{N}\)}
            \For{\(r=1,\ldots,R\)}
                \State Generate a right-censored sample from the Weibull mixture cure model
                \State Jointly estimate \(a\), \(b\), and \(p\) by maximum likelihood, obtaining \((\widehat a,\widehat b,\widehat p)\)
                \If{the fitted model is numerically valid}
                    \For{each \(\delta\in\mathcal{D}\)}
                        \State Calculate and record
                        \[
                        \widehat t_P(\delta)
                        =
                        \widehat b
                        \left[
                        -\log\left\{
                        \frac{\delta}{1-\widehat p}
                        \right\}
                        \right]^{1/\widehat a}
                        \]
                    \EndFor
                    \For{each \(\varepsilon\in\mathcal{E}\)}
                        \State Calculate and record
                        \[
                        \widehat t_R(\varepsilon)
                        =
                        \widehat b[-\log(\varepsilon)]^{1/\widehat a}
                        \]
                    \EndFor
                \Else
                    \State Mark the replication as invalid
                \EndIf
            \EndFor
            \State For each tolerance, compute the Monte Carlo mean and sample standard deviation using the valid replications
        \EndFor
    \EndFor
\EndFor
\end{algorithmic}
\end{algorithm}

Table~\ref{tab:pdc_simulation} presents the Monte Carlo means and sample standard deviations of the minimum sufficient follow-up times estimated using the PDC.

\begin{table}[ht!]
\centering
\scriptsize
\caption{Monte Carlo means and sample standard deviations of the estimated minimum sufficient follow-up times obtained using the PDC. Results are based on \(R=1000\) replications with \(a\), \(b\), and \(p\) estimated jointly.}
\label{tab:pdc_simulation}
\setlength{\tabcolsep}{4pt}
\renewcommand{\arraystretch}{1.0}
\begin{adjustbox}{max width=\textwidth}
\begin{tabular}{ccccccccc}
\toprule
\multirow[c]{2}{*}{$\lambda$}
& \multirow[c]{2}{*}{$p$}
& \multirow[c]{2}{*}{$n$}
& \multicolumn{2}{c}{$\delta=0.050$}
& \multicolumn{2}{c}{$\delta=0.025$}
& \multicolumn{2}{c}{$\delta=0.010$} \\
\cmidrule(lr){4-5}
\cmidrule(lr){6-7}
\cmidrule(lr){8-9}
& & &
$\widehat{t}_P(0.050)$ & SD &
$\widehat{t}_P(0.025)$ & SD &
$\widehat{t}_P(0.010)$ & SD \\
\midrule

\multirow[c]{15}{*}{$2$}
& \multirow[c]{5}{*}{$0.25$}
& 100   & 2.964 & 1.444 & 3.446 & 1.724 & 4.040 & 2.092 \\
& & 500   & 3.088 & 0.971 & 3.591 & 1.138 & 4.207 & 1.354 \\
& & 1000  & 3.018 & 0.764 & 3.509 & 0.892 & 4.112 & 1.057 \\
& & 5000  & 2.964 & 0.343 & 3.449 & 0.401 & 4.043 & 0.475 \\
& & 10000 & 2.934 & 0.244 & 3.416 & 0.284 & 4.005 & 0.337 \\
\cmidrule(lr){2-9}

& \multirow[c]{5}{*}{$0.50$}
& 100   & 3.421 & 3.102 & 4.051 & 3.751 & 4.833 & 4.615 \\
& & 500   & 3.111 & 1.702 & 3.690 & 2.007 & 4.398 & 2.404 \\
& & 1000  & 2.933 & 1.267 & 3.482 & 1.487 & 4.152 & 1.772 \\
& & 5000  & 2.662 & 0.393 & 3.170 & 0.465 & 3.786 & 0.559 \\
& & 10000 & 2.635 & 0.281 & 3.140 & 0.334 & 3.751 & 0.402 \\
\cmidrule(lr){2-9}

& \multirow[c]{5}{*}{$0.75$}
& 100   & 6.208 & 19.843 & 7.716 & 26.582 & 9.734 & 36.508 \\
& & 500   & 3.547 & 4.422 & 4.350 & 5.252 & 5.328 & 6.335 \\
& & 1000  & 3.025 & 3.294 & 3.745 & 3.900 & 4.613 & 4.689 \\
& & 5000  & 2.167 & 0.592 & 2.743 & 0.722 & 3.424 & 0.890 \\
& & 10000 & 2.107 & 0.357 & 2.672 & 0.439 & 3.338 & 0.545 \\
\midrule

\multirow[c]{15}{*}{$3.5$}
& \multirow[c]{5}{*}{$0.25$}
& 100   & 3.051 & 0.910 & 3.554 & 1.108 & 4.172 & 1.371 \\
& & 500   & 2.962 & 0.411 & 3.450 & 0.499 & 4.047 & 0.613 \\
& & 1000  & 2.934 & 0.269 & 3.416 & 0.326 & 4.006 & 0.400 \\
& & 5000  & 2.922 & 0.114 & 3.401 & 0.138 & 3.987 & 0.170 \\
& & 10000 & 2.918 & 0.078 & 3.397 & 0.095 & 3.984 & 0.117 \\
\cmidrule(lr){2-9}

& \multirow[c]{5}{*}{$0.50$}
& 100   & 3.001 & 1.858 & 3.580 & 2.290 & 4.294 & 2.868 \\
& & 500   & 2.677 & 0.477 & 3.192 & 0.592 & 3.817 & 0.745 \\
& & 1000  & 2.638 & 0.296 & 3.144 & 0.366 & 3.757 & 0.457 \\
& & 5000  & 2.615 & 0.122 & 3.117 & 0.151 & 3.724 & 0.190 \\
& & 10000 & 2.615 & 0.090 & 3.117 & 0.112 & 3.724 & 0.141 \\
\cmidrule(lr){2-9}

& \multirow[c]{5}{*}{$0.75$}
& 100   & 3.926 & 8.736 & 4.922 & 11.128 & 6.182 & 14.436 \\
& & 500   & 2.158 & 0.717 & 2.739 & 0.923 & 3.428 & 1.199 \\
& & 1000  & 2.114 & 0.358 & 2.685 & 0.465 & 3.360 & 0.610 \\
& & 5000  & 2.067 & 0.137 & 2.625 & 0.179 & 3.284 & 0.234 \\
& & 10000 & 2.063 & 0.101 & 2.620 & 0.131 & 3.276 & 0.172 \\
\midrule

\multirow[c]{15}{*}{$5$}
& \multirow[c]{5}{*}{$0.25$}
& 100   & 2.950 & 0.544 & 3.434 & 0.676 & 4.027 & 0.852 \\
& & 500   & 2.918 & 0.223 & 3.397 & 0.274 & 3.983 & 0.343 \\
& & 1000  & 2.922 & 0.156 & 3.402 & 0.192 & 3.988 & 0.240 \\
& & 5000  & 2.913 & 0.069 & 3.390 & 0.085 & 3.974 & 0.107 \\
& & 10000 & 2.916 & 0.047 & 3.395 & 0.058 & 3.981 & 0.073 \\
\cmidrule(lr){2-9}

& \multirow[c]{5}{*}{$0.50$}
& 100   & 2.692 & 1.150 & 3.211 & 1.474 & 3.844 & 1.922 \\
& & 500   & 2.628 & 0.252 & 3.132 & 0.318 & 3.742 & 0.408 \\
& & 1000  & 2.615 & 0.162 & 3.117 & 0.203 & 3.725 & 0.260 \\
& & 5000  & 2.617 & 0.073 & 3.119 & 0.092 & 3.726 & 0.118 \\
& & 10000 & 2.615 & 0.053 & 3.116 & 0.067 & 3.723 & 0.085 \\
\cmidrule(lr){2-9}

& \multirow[c]{5}{*}{$0.75$}
& 100   & 2.378 & 3.472 & 3.026 & 4.614 & 3.816 & 6.249 \\
& & 500   & 2.074 & 0.272 & 2.633 & 0.358 & 3.295 & 0.479 \\
& & 1000  & 2.065 & 0.180 & 2.623 & 0.238 & 3.281 & 0.321 \\
& & 5000  & 2.062 & 0.083 & 2.620 & 0.109 & 3.277 & 0.145 \\
& & 10000 & 2.060 & 0.057 & 2.615 & 0.075 & 3.270 & 0.100 \\

\bottomrule
\end{tabular}
\end{adjustbox}
\end{table}

For \(\lambda=2\), all mean PDC requirements exceed the maximum potential follow-up of \(2\). Joint estimation is particularly unstable for small samples and high cure fractions. For example, when \(p=0.75\) and \(n=100\), the mean estimates are \(6.208\), \(7.716\), and \(9.734\) for \(\delta=0.050\), \(0.025\), and \(0.010\), respectively, with corresponding standard deviations of \(19.843\), \(26.582\), and \(36.508\). This instability reflects the difficulty of jointly estimating the susceptible survival distribution and cure fraction when the observation window is short and relatively few susceptible failures are observed.

For \(\lambda=3.5\), the relationship between the estimated minimum requirement and the available follow-up depends on the tolerance and cure fraction. Once joint estimation becomes stable, the mean PDC times for \(\delta=0.050\) are below \(3.5\). At \(n=10000\), the mean requirements for \(\delta=0.010\) are \(3.984\), \(3.724\), and \(3.276\) for \(p=0.25\), \(0.50\), and \(0.75\), respectively. Thus, under this more stringent tolerance, the requirement remains above the censoring limit for \(p=0.25\) and \(p=0.50\), but falls below it for \(p=0.75\). The small-sample scenario with \(p=0.75\) remains highly variable; at \(n=100\), the standard deviations range from \(8.736\) to \(14.436\) across the three tolerances.

For \(\lambda=5\), all mean PDC times are shorter than the maximum available follow-up. The estimates approach their theoretical targets as the sample size increases, and their standard deviations decrease markedly with \(n\). At \(n=10000\), the mean estimates are \(2.916\), \(3.395\), and \(3.981\) for \(p=0.25\); \(2.615\), \(3.116\), and \(3.723\) for \(p=0.50\); and \(2.060\), \(2.615\), and \(3.270\) for \(p=0.75\). The corresponding standard deviations range from \(0.047\) to \(0.100\). For moderate and large sample sizes, the mean PDC times decrease with the cure fraction, as expected from the absolute-distance formulation because the corresponding susceptible-scale tolerance is \(\delta/(1-p)\).

Overall, the PDC results show that more stringent tolerances require longer follow-up and that joint estimation can produce substantial uncertainty under short follow-up, small sample sizes, and high cure fractions. As the sample size and available observation period increase, the estimates become more stable and approach the corresponding theoretical minimum follow-up times. Table~\ref{tab:rsc_simulation} presents the Monte Carlo means and sample standard deviations of the minimum sufficient follow-up times estimated using the RSC.

\begin{table}[ht!]
\centering
\scriptsize
\caption{Monte Carlo means and sample standard deviations of the estimated minimum sufficient follow-up times obtained using the RSC. Results are based on \(R=1000\) replications with \(a\), \(b\), and \(p\) estimated jointly.}
\label{tab:rsc_simulation}
\setlength{\tabcolsep}{4pt}
\renewcommand{\arraystretch}{1.0}
\begin{adjustbox}{max width=\textwidth}
\begin{tabular}{ccccccccc}
\toprule
\multirow[c]{2}{*}{$\lambda$}
& \multirow[c]{2}{*}{$p$}
& \multirow[c]{2}{*}{$n$}
& \multicolumn{2}{c}{$\varepsilon=0.050$}
& \multicolumn{2}{c}{$\varepsilon=0.025$}
& \multicolumn{2}{c}{$\varepsilon=0.010$} \\
\cmidrule(lr){4-5}
\cmidrule(lr){6-7}
\cmidrule(lr){8-9}
& & &
$\widehat{t}_R(0.050)$ & SD &
$\widehat{t}_R(0.025)$ & SD &
$\widehat{t}_R(0.010)$ & SD \\
\midrule

\multirow[c]{15}{*}{$2$}
& \multirow[c]{5}{*}{$0.25$}
& 100   & 3.113 & 1.371 & 3.583 & 1.662 & 4.166 & 2.039 \\
& & 500   & 3.248 & 0.890 & 3.739 & 1.065 & 4.343 & 1.289 \\
& & 1000  & 3.194 & 0.697 & 3.673 & 0.832 & 4.263 & 1.003 \\
& & 5000  & 3.159 & 0.315 & 3.631 & 0.376 & 4.211 & 0.453 \\
& & 10000 & 3.134 & 0.225 & 3.601 & 0.267 & 4.177 & 0.322 \\
\cmidrule(lr){2-9}

& \multirow[c]{5}{*}{$0.50$}
& 100   & 3.742 & 2.987 & 4.345 & 3.655 & 5.101 & 4.536 \\
& & 500   & 3.535 & 1.607 & 4.082 & 1.927 & 4.759 & 2.336 \\
& & 1000  & 3.385 & 1.196 & 3.900 & 1.426 & 4.537 & 1.721 \\
& & 5000  & 3.157 & 0.396 & 3.628 & 0.471 & 4.209 & 0.567 \\
& & 10000 & 3.133 & 0.287 & 3.601 & 0.342 & 4.176 & 0.411 \\
\cmidrule(lr){2-9}

& \multirow[c]{5}{*}{$0.75$}
& 100   & 6.840 & 19.732 & 8.309 & 26.523 & 10.281 & 36.469 \\
& & 500   & 4.434 & 4.253 & 5.160 & 5.112 & 6.068 & 6.221 \\
& & 1000  & 4.000 & 3.172 & 4.637 & 3.802 & 5.429 & 4.612 \\
& & 5000  & 3.228 & 0.698 & 3.714 & 0.831 & 4.312 & 1.001 \\
& & 10000 & 3.167 & 0.441 & 3.640 & 0.525 & 4.223 & 0.632 \\
\midrule

\multirow[c]{15}{*}{$3.5$}
& \multirow[c]{5}{*}{$0.25$}
& 100   & 3.231 & 0.884 & 3.721 & 1.088 & 4.327 & 1.357 \\
& & 500   & 3.160 & 0.407 & 3.634 & 0.497 & 4.218 & 0.613 \\
& & 1000  & 3.135 & 0.268 & 3.603 & 0.326 & 4.179 & 0.402 \\
& & 5000  & 3.124 & 0.114 & 3.589 & 0.139 & 4.161 & 0.172 \\
& & 10000 & 3.121 & 0.079 & 3.586 & 0.096 & 4.158 & 0.118 \\
\cmidrule(lr){2-9}

& \multirow[c]{5}{*}{$0.50$}
& 100   & 3.464 & 1.862 & 4.010 & 2.306 & 4.692 & 2.894 \\
& & 500   & 3.178 & 0.532 & 3.656 & 0.651 & 4.246 & 0.807 \\
& & 1000  & 3.138 & 0.332 & 3.607 & 0.404 & 4.184 & 0.497 \\
& & 5000  & 3.117 & 0.138 & 3.581 & 0.169 & 4.152 & 0.208 \\
& & 10000 & 3.117 & 0.103 & 3.581 & 0.126 & 4.152 & 0.155 \\
\cmidrule(lr){2-9}

& \multirow[c]{5}{*}{$0.75$}
& 100   & 4.922 & 8.882 & 5.844 & 11.319 & 7.038 & 14.670 \\
& & 500   & 3.235 & 0.950 & 3.726 & 1.165 & 4.334 & 1.446 \\
& & 1000  & 3.189 & 0.516 & 3.669 & 0.629 & 4.261 & 0.777 \\
& & 5000  & 3.128 & 0.201 & 3.595 & 0.245 & 4.170 & 0.302 \\
& & 10000 & 3.122 & 0.148 & 3.588 & 0.180 & 4.161 & 0.222 \\
\midrule

\multirow[c]{15}{*}{$5$}
& \multirow[c]{5}{*}{$0.25$}
& 100   & 3.147 & 0.557 & 3.617 & 0.693 & 4.197 & 0.872 \\
& & 500   & 3.120 & 0.231 & 3.585 & 0.284 & 4.157 & 0.354 \\
& & 1000  & 3.125 & 0.161 & 3.590 & 0.198 & 4.163 & 0.247 \\
& & 5000  & 3.115 & 0.072 & 3.579 & 0.088 & 4.148 & 0.110 \\
& & 10000 & 3.119 & 0.049 & 3.584 & 0.060 & 4.156 & 0.075 \\
\cmidrule(lr){2-9}

& \multirow[c]{5}{*}{$0.50$}
& 100   & 3.188 & 1.202 & 3.672 & 1.532 & 4.271 & 1.983 \\
& & 500   & 3.129 & 0.299 & 3.596 & 0.369 & 4.171 & 0.461 \\
& & 1000  & 3.117 & 0.192 & 3.582 & 0.236 & 4.154 & 0.294 \\
& & 5000  & 3.119 & 0.086 & 3.583 & 0.106 & 4.154 & 0.133 \\
& & 10000 & 3.116 & 0.063 & 3.580 & 0.077 & 4.150 & 0.096 \\
\cmidrule(lr){2-9}

& \multirow[c]{5}{*}{$0.75$}
& 100   & 3.470 & 3.807 & 4.038 & 4.994 & 4.757 & 6.676 \\
& & 500   & 3.139 & 0.421 & 3.608 & 0.518 & 4.186 & 0.647 \\
& & 1000  & 3.128 & 0.282 & 3.595 & 0.350 & 4.169 & 0.439 \\
& & 5000  & 3.123 & 0.127 & 3.589 & 0.156 & 4.162 & 0.195 \\
& & 10000 & 3.118 & 0.088 & 3.582 & 0.108 & 4.153 & 0.135 \\

\bottomrule
\end{tabular}
\end{adjustbox}
\end{table}

For \(\lambda=2\), all mean RSC requirements exceed the maximum potential follow-up of \(2\). Joint estimation is especially unstable for small samples with large cure fractions. When \(p=0.75\) and \(n=100\), the mean estimates are \(6.840\), \(8.309\), and \(10.281\) for \(\varepsilon=0.050\), \(0.025\), and \(0.010\), respectively, with corresponding standard deviations of \(19.732\), \(26.523\), and \(36.469\). These results show that insufficient follow-up affects not only the estimated cure fraction but also the precision with which the susceptible survival distribution can be recovered.

For \(\lambda=3.5\), the theoretical requirements are approximately \(3.117\), \(3.581\), and \(4.152\) for \(\varepsilon=0.050\), \(0.025\), and \(0.010\), respectively. Hence, only the first target lies below the available follow-up limit. At \(n=10000\), the estimated means range across the three cure fractions from \(3.117\) to \(3.122\), from \(3.581\) to \(3.588\), and from \(4.152\) to \(4.161\) for the three respective tolerances. The principal exception again occurs in the smallest samples with \(p=0.75\); at \(n=100\), the standard deviations are \(8.882\), \(11.319\), and \(14.670\).

For \(\lambda=5\), all mean RSC times are shorter than the maximum available follow-up. For moderate and large sample sizes, the estimates approach the theoretical targets of approximately \(3.117\), \(3.581\), and \(4.152\), and their standard deviations decrease steadily with the sample size. At \(n=10000\), the mean estimates differ by no more than \(0.006\) across the three cure fractions for any tolerance, while the standard deviations range from \(0.049\) to \(0.135\).

Overall, the RSC estimates approach their theoretical targets as the amount of information increases. Although the theoretical RSC target is independent of the cure fraction, the finite-sample precision of its estimator is not. Larger cure fractions reduce the number of susceptible individuals and observed events, increasing uncertainty in \(\widehat a\) and \(\widehat b\), particularly when the observation window is short.

%\newpage
%\clearpage

\section{Applications}
\label{sec:applications}

This section illustrates the proposed framework using two survival datasets with evidence of a cure fraction. The first application considers a population-based cohort of patients with prostate cancer, whereas the second involves patients with triple-negative breast cancer. In each application, we first evaluate follow-up sufficiency using the PFST and then use the PDC and RSC to estimate the minimum follow-up required at different tolerance levels. Together, these procedures assess whether the available follow-up is sufficient and quantify the additional observation time required to support reliable estimation of the cure fraction.

\subsection{Prostate cancer data}
\label{sec:aplicacao}

To illustrate the proposed methodology, we analyze data from a population-based
cancer registry comprising patients diagnosed with prostate cancer
between 2016 and 2022. The event of interest was prostate cancer-specific
death, whereas patients who remained alive at the end of follow-up or
died from other causes were treated as censored observations.

The dataset comprises \(n=30,743\) patients, of whom \(2,350\) experienced
the event of interest and \(28,393\) were censored, yielding a
censoring proportion of \(92.36\%\). The observed follow-up times ranged
from \(0.0027\) years to
\(\tmax=6.90\) years, and the latest observed event occurred at
\(\tilde{t}_{K}=6.57\) years.

We first evaluated follow-up sufficiency using the PFST under the
hypotheses \(H_0:\tau_{F_0}\leq\tau_C\) and
\(H_1:\tau_{F_0}>\tau_C\). Rejection of \(H_0\) indicates evidence that
the available follow-up is insufficient for reliable estimation of the
cure fraction.

We then fitted the Weibull mixture cure model introduced in Section~3
to the data. The maximum likelihood estimates were
\(\widehat{a}=1.1053\), \(\widehat{b}=5.9885\), and
\(\widehat{p}=0.7691\), corresponding to an estimated cure fraction of
\(76.91\%\). The fitted population survival function was

\[
\widehat{S}_p(t)
=
0.7691+
\left(1-0.7691\right)
\exp\left\{
-\left(\frac{t}{5.9885}\right)^{1.1053}
\right\}.
\]

Figure~\ref{fig:aplicacao_prostata}
(\subref{fig:km_weibull_cura_prostata}) shows the Kaplan-Meier
estimator as a black step curve and the population survival curve
estimated by the Weibull cure model as a solid red curve. The estimated
cure fraction, \(\widehat{p}=0.7691\), is represented by a dashed blue
horizontal line, whereas the terminal Kaplan-Meier estimate,
\(\widehat{p}_{KM}=0.8259\), is represented by a dotted gray horizontal
line. Figure~\ref{fig:aplicacao_prostata}
(\subref{fig:hist_parametrico_prostata}) shows the centered bootstrap
distribution used by the PFST together with its asymptotic normal
approximation. The gray histogram represents the centered nonparametric
bootstrap distribution, whereas the solid green curve represents the
asymptotic distribution obtained from the joint influence-function
approximation. The solid red vertical line indicates the observed PFST
statistic, whereas the solid black and blue vertical lines indicate the
bootstrap and asymptotic critical values, respectively.

\begin{figure}[!ht]
\centering
\begin{subfigure}[t]{0.49\textwidth}
\centering
\includegraphics[width=\textwidth]{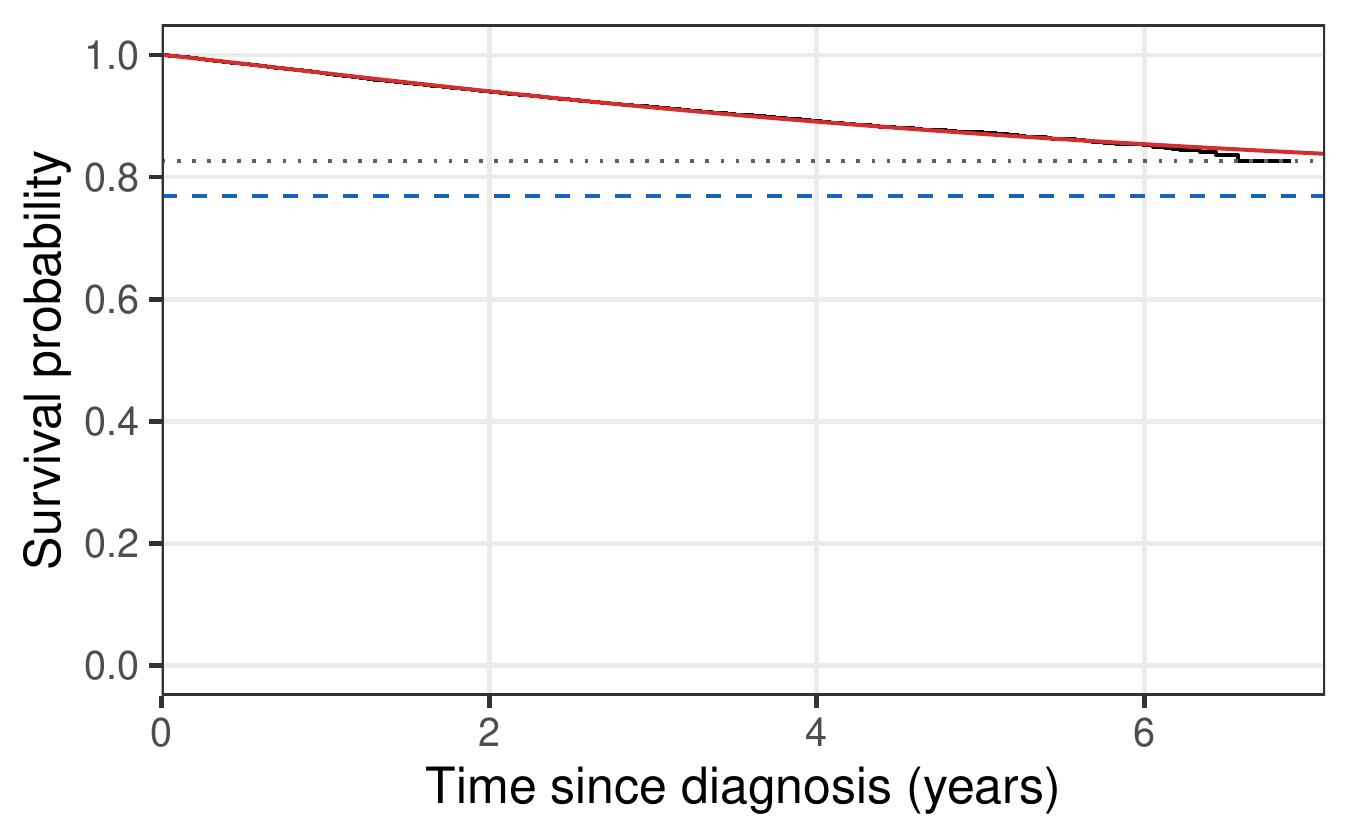}
\caption{Kaplan-Meier estimator and fitted Weibull cure model.}
\label{fig:km_weibull_cura_prostata}
\end{subfigure}
\hfill
\begin{subfigure}[t]{0.49\textwidth}
\centering
\includegraphics[width=\textwidth]{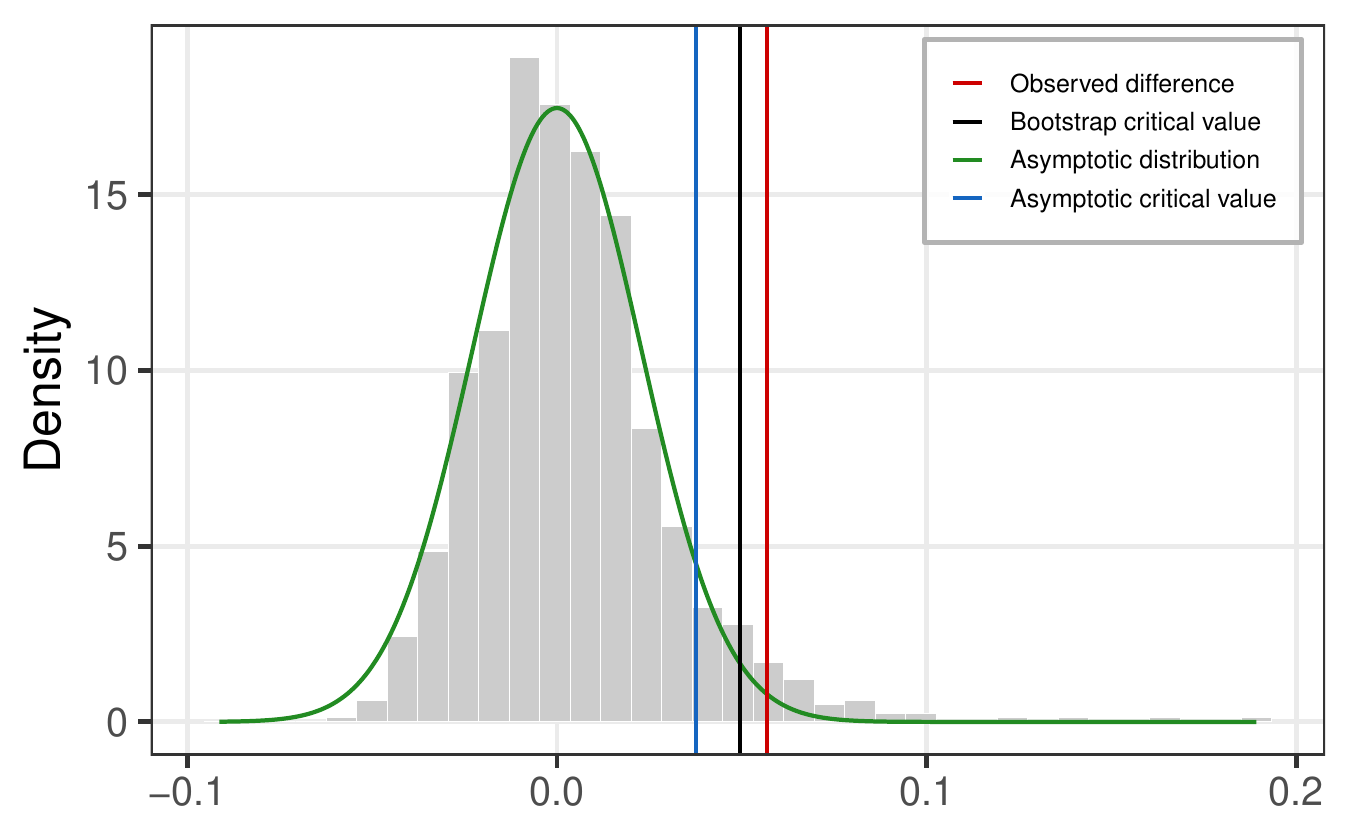}
\caption{Centered bootstrap distribution and asymptotic approximation of
the PFST statistic.}
\label{fig:hist_parametrico_prostata}
\end{subfigure}
\caption{Graphical assessment of the Weibull cure model fit and the PFST
for the prostate cancer data.}
\label{fig:aplicacao_prostata}
\end{figure}

The terminal Kaplan-Meier estimate was
\(\widehat{p}_{KM}=0.8259\), whereas the cure fraction estimated by the
model was \(\widehat{p}=0.7691\). Consequently, the observed PFST
statistic was

\[
\Tparam
=
\widehat{p}_{KM}-\widehat{p}
=
0.0569.
\]

The bootstrap PFST decision was based on a centered nonparametric
bootstrap procedure with \(1000\) valid replications. The estimated
critical value at the \(5\%\) significance level was \(0.0495\), and the
bootstrap \(p\)-value was \(0.035\). The influence-function approximation
yielded \(\widehat{\tau}^{2}=16.0296\), corresponding to an estimated
standard error of \(0.0228\), an asymptotic critical value of \(0.0376\),
a \(Z\)-statistic of \(2.4898\), and an asymptotic \(p\)-value of
\(0.0064\). As shown in Figure~\ref{fig:aplicacao_prostata}
(\subref{fig:hist_parametrico_prostata}), the red vertical line
corresponding to the observed PFST statistic lies to the right of
both the black bootstrap critical value and the blue asymptotic critical
value. Since \(0.0569>0.0495\) and \(0.0569>0.0376\), both procedures
rejected \(H_0\). Thus, the bootstrap and asymptotic procedures provide
concordant evidence that the available follow-up is insufficient for
reliable estimation of the cure fraction. Although the asymptotic
procedure produced a smaller critical value and \(p\)-value, the
inferential conclusion was unchanged.

Given the evidence of insufficient follow-up, we applied the PDC to
estimate the minimum follow-up required at different
tolerance levels. This criterion identifies the minimum time at which
the population survival curve estimated by the Weibull cure model
approaches the cure fraction within a prespecified tolerance.
Figure~\ref{fig:criterios_prostata}
(\subref{fig:plateau_distance_application}) shows the Kaplan-Meier
estimator as a black step curve, the fitted population survival curve as
a solid red curve, and the estimated cure fraction as a dotted blue
horizontal line. The dashed vertical lines, distinguished by color,
represent the minimum follow-up times obtained for
\(\delta\in\{0.100,\allowbreak 0.050,\allowbreak 0.025,\allowbreak
0.010,\allowbreak 0.005\}\).
Figure~\ref{fig:criterios_prostata}
(\subref{fig:residual_survival_application}) uses the same graphical
representation, with colored dashed vertical lines indicating the
minimum follow-up times obtained by the RSC for
\(\varepsilon\in\{0.100,\allowbreak 0.050,\allowbreak 0.025,\allowbreak
0.010,\allowbreak 0.005\}\).

\begin{figure}[!ht]
\centering
\begin{subfigure}[t]{0.49\textwidth}
\centering
\includegraphics[width=\textwidth]{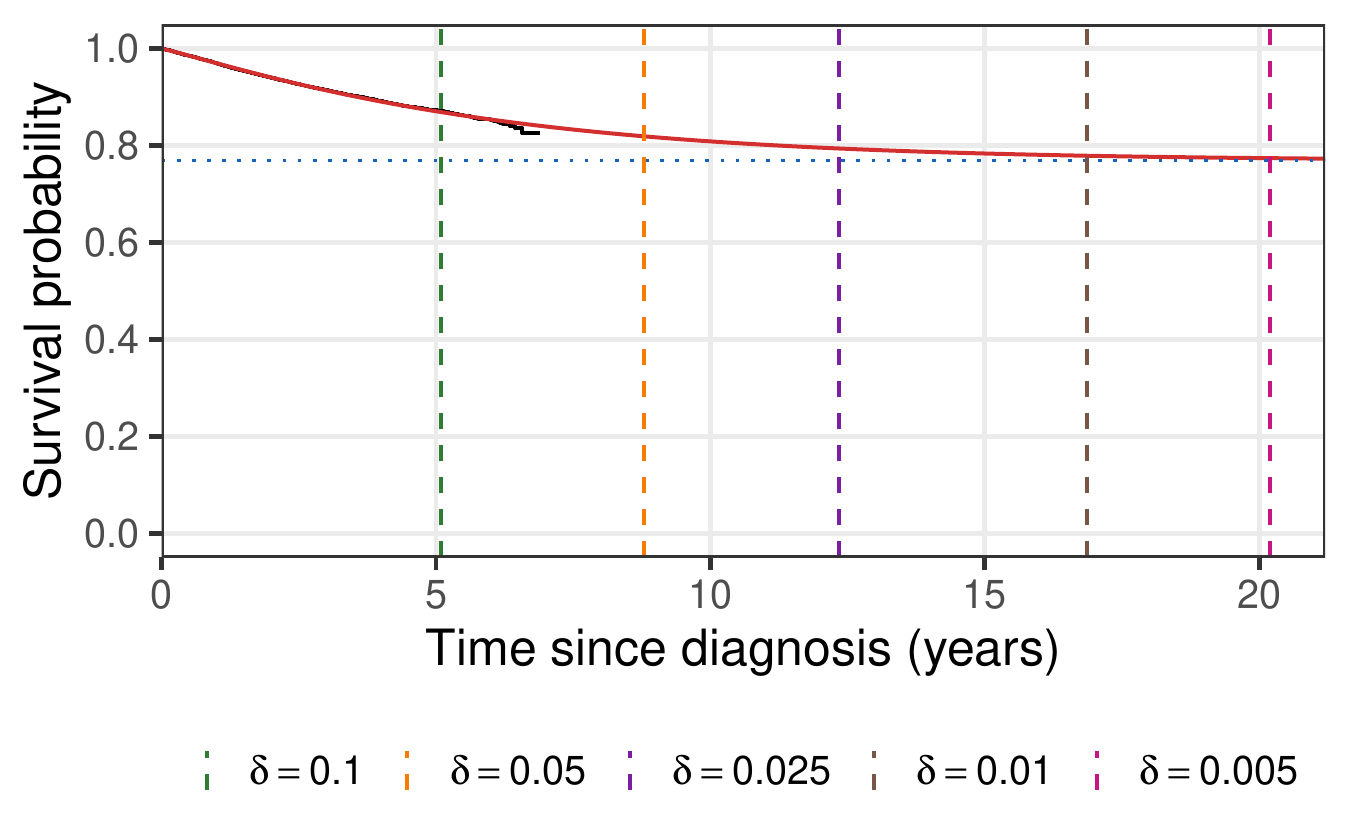}
\caption{Minimum follow-up times obtained by the PDC for different
tolerance levels \(\delta\).}
\label{fig:plateau_distance_application}
\end{subfigure}
\hfill
\begin{subfigure}[t]{0.49\textwidth}
\centering
\includegraphics[width=\textwidth]{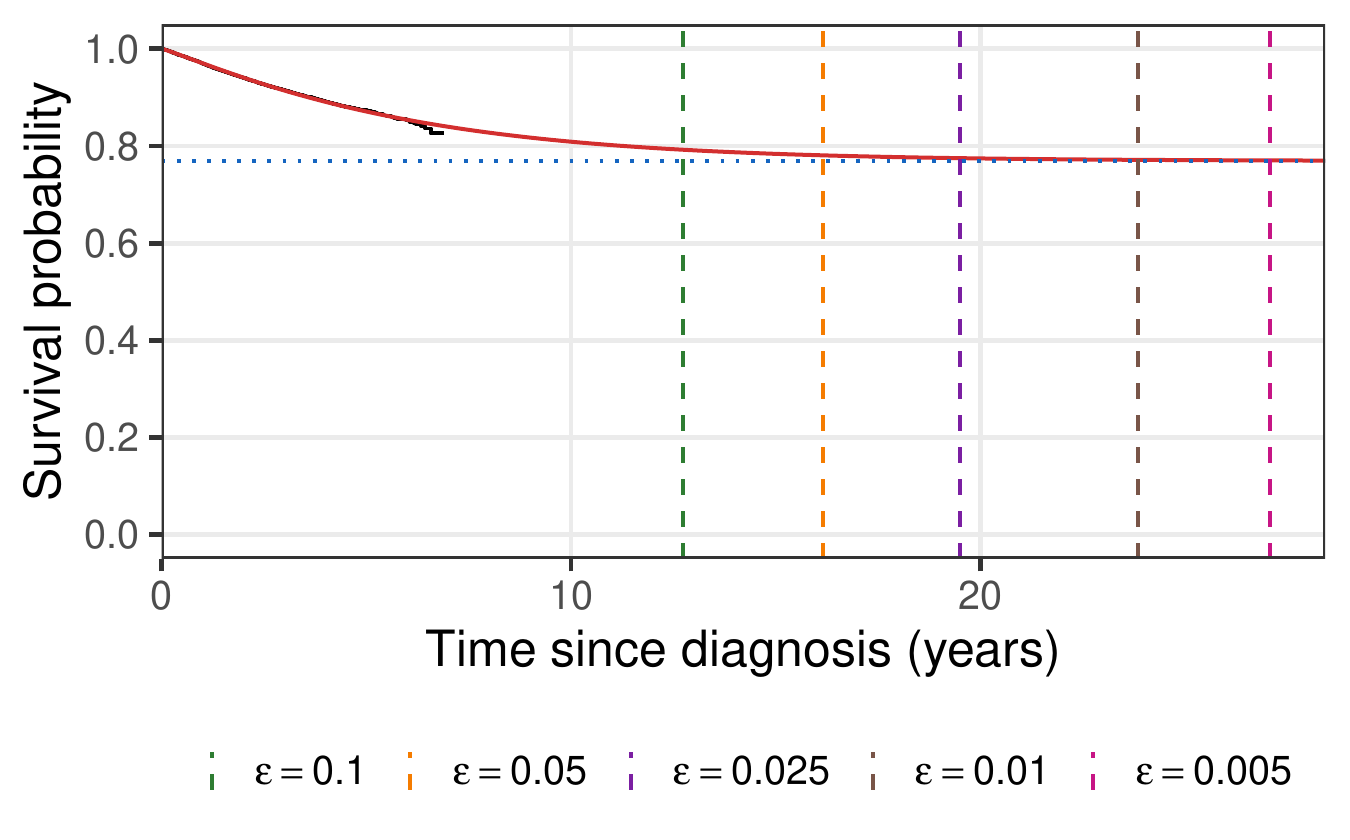}
\caption{Minimum follow-up times obtained by the RSC for different
tolerance levels \(\varepsilon\).}
\label{fig:residual_survival_application}
\end{subfigure}
\caption{Minimum follow-up times estimated by the PDC and RSC for the
prostate cancer data. In both panels, the black step curve represents
the Kaplan-Meier estimator, the solid red curve represents the fitted
Weibull cure model, and the dotted blue horizontal line represents the
estimated cure fraction.}
\label{fig:criterios_prostata}
\end{figure}

Table~\ref{tab:framework_summary_prostate} summarizes the complementary
roles of the PFST, PDC, and RSC in evaluating follow-up sufficiency. We
calculated the time difference as the estimated minimum follow-up
minus the maximum observed follow-up time of \(6.90\) years. Thus,
negative values indicate that the maximum observed follow-up exceeds the
estimated minimum, whereas positive values represent the additional
follow-up time required.

\begin{table}[!ht]
\centering
\caption{Summary of the follow-up sufficiency assessment for the
prostate cancer data.}
\label{tab:framework_summary_prostate}
\renewcommand{\arraystretch}{1.20}
\small
\setlength{\tabcolsep}{8pt}

\begin{tabular}{llcc}
\toprule
\textbf{Method} &
\textbf{Tolerance} &
\makecell{\textbf{Observed or minimum}\\
\textbf{follow-up (years)}} &
\makecell{\textbf{Follow-up assessment}\\
\textbf{and time difference (years)}}\\
\midrule

PFST &
\makecell{Bootstrap: \(p=0.035\)\\Asymptotic: \(p=0.0064\)} &
\makecell{\textit{Maximum observed}\\follow-up: 6.90} &
Insufficient\\

\midrule

PDC &
\(\delta=0.100\) &
5.10 &
Sufficient (\(-1.80\))\\

&
\(\delta=0.050\) &
8.80 &
Insufficient (\(+1.90\))\\

&
\(\delta=0.025\) &
12.34 &
Insufficient (\(+5.44\))\\

&
\(\delta=0.010\) &
16.86 &
Insufficient (\(+9.96\))\\

&
\(\delta=0.005\) &
20.19 &
Insufficient (\(+13.30\))\\

\midrule

RSC &
\(\varepsilon=0.100\) &
12.74 &
Insufficient (\(+5.84\))\\

&
\(\varepsilon=0.050\) &
16.16 &
Insufficient (\(+9.26\))\\

&
\(\varepsilon=0.025\) &
19.51 &
Insufficient (\(+12.61\))\\

&
\(\varepsilon=0.010\) &
23.84 &
Insufficient (\(+16.95\))\\

&
\(\varepsilon=0.005\) &
27.07 &
Insufficient (\(+20.17\))\\

\bottomrule
\end{tabular}
\end{table}

The PFST provided evidence of insufficient follow-up under both the
bootstrap and asymptotic procedures. Consistent with this result, the PDC indicated
that the maximum observed follow-up satisfied only the least stringent
tolerance level, \(\delta=0.100\), exceeding the corresponding minimum
of \(5.10\) years by \(1.80\) years. A total of \(3,082\) patients had
observed follow-up times at or beyond this threshold, including \(49\)
events and \(3,033\) censored observations. For
\(\delta\leq0.050\), the PDC required between \(1.90\) and \(13.30\)
additional years beyond the maximum observed follow-up. None of the RSC
thresholds were attained, and the estimated additional follow-up ranged
from \(5.84\) to \(20.17\) years. Therefore, the available follow-up is
adequate only under a very lenient population-level definition of
proximity to the cure plateau and remains insufficient under moderate
or stringent PDC tolerances and under all RSC tolerances considered.
The value of \(6.90\) years should be interpreted as the longest
individual observation rather than as a common follow-up duration for
the entire cohort.

A direct comparison of the two criteria using identical numerical
tolerance values shows that the RSC produced minimum
follow-up times approximately \(6.9\)--\(7.6\) years longer than those
obtained by the PDC. This difference is expected because equal numerical
values of \(\delta\) and \(\varepsilon\) correspond to different
practical definitions of follow-up sufficiency rather than equivalent
levels of stringency. This apparent discrepancy disappears once the
tolerances are matched according to the theoretical relationship
between the two criteria.

The PDC and RSC are mathematically equivalent when their tolerances are
related by \(\varepsilon=\delta/(1-\widehat{p})\). In the application to
the prostate cancer data, the estimated cure fraction was
\(\widehat{p}=0.7691\), so that \(1-\widehat{p}=0.2309\). The equivalent
tolerance pairs were
\((\delta,\varepsilon)=(0.100,0.4330)\),
\((0.050,0.2165)\),
\((0.025,0.1083)\),
\((0.010,0.0433)\), and
\((0.005,0.0217)\).
Substituting each pair into the expressions for
\(\widehat{t}_P(\delta)\) and \(\widehat{t}_R(\varepsilon)\) yielded the
same minimum follow-up times: \(5.10\), \(8.80\), \(12.34\), \(16.86\),
and \(20.19\) years, respectively. Consequently, the choice between the
two criteria should be guided by the practical interpretation of the
tolerance parameter rather than by numerical comparisons using
identical tolerance values.

This verification confirms that the differences observed when equal
nominal values of \(\delta\) and \(\varepsilon\) are used arise not
from the expressions defining the criteria but from the fact that numerically
equal tolerances represent different conditions. Therefore, the longer
times obtained by the RSC when \(\delta=\varepsilon\)
reflect a more stringent requirement for the residual survival
probability of susceptible individuals rather than a lack of
equivalence between the criteria.

In the present application, the minimum follow-up times were obtained
from a Weibull mixture cure model, whereas \citet{tai2005minimum}
proposed a criterion based on fitting a lognormal distribution to the
survival times of patients who died from prostate cancer. Their
threshold corresponds to the time at which the fitted lognormal
distribution is virtually exhausted, leaving less than \(2.25\%\) of
its upper tail unobserved. Although the underlying parametric models
differ, Tai et al.'s criterion is conceptually equivalent to the RSC
because both are based on controlling the residual survival probability
among susceptible individuals. Using the notation adopted in this
paper, their criterion can be expressed as \(S_0(t)\leq0.0225\),
implying that \(t_{\mathrm{Tai}}\approx t_R(0.0225)\).

The estimated follow-up times are also consistent with the results reported
by \citet{tai2005minimum}, who obtained a threshold of \(24.6\) years
for statistical cure among prostate cancer patients younger than
60 years diagnosed between 1973 and 1977. This value is close
with the more conservative estimates obtained by the RSC, particularly
\(\widehat{t}_R(0.010)=23.84\) years and
\(\widehat{t}_R(0.005)=27.07\) years.

Although the comparison with \citet{tai2005minimum} establishes the
methodological consistency of the RSC with an existing
criterion for minimum follow-up assessment, it is also important to
examine whether the estimated follow-up times agree with evidence from
long-term clinical studies. Table~\ref{tab:clinical_comparison}
summarizes this comparison for representative prostate cancer cohorts.

\begin{table}[ht!]
\centering
\caption{Comparison of the minimum follow-up times estimated by the
proposed framework with evidence from long-term prostate cancer studies.}
\label{tab:clinical_comparison}
\begin{tabular}{p{3.2cm}p{5.3cm}p{2cm}p{4.5cm}}
\hline
\textbf{Study} &
\textbf{Main finding} &
\textbf{Maximum follow-up} &
\textbf{Consistency with the proposed framework}\\
\hline

\citet{klotz2015long}
&
Active surveillance appeared safe over long-term follow-up, with
clinically meaningful outcomes reported after 15 years.
&
19.8 years
&
Consistent with the RSC estimates,
\(\widehat{t}_R(0.050)=16.16\) years and
\(\widehat{t}_R(0.025)=19.51\) years.
\\

\citet{de2023detailed}
&
Metastatic disease and prostate cancer-specific mortality continued to
evolve throughout long-term follow-up.
&
21 years
&
Consistent with the conservative estimates
\(\widehat{t}_P(0.005)=20.19\) years and
\(\widehat{t}_R(0.025)=19.51\) years.
\\

\citet{fraanlund2022results}
&
Longer follow-up changed previous estimates of the number needed to
diagnose.
&
22 years
&
Compatible with the more conservative minimum follow-up times
estimated by the PDC and RSC.
\\

\hline
\end{tabular}
\end{table}

Table~\ref{tab:clinical_comparison} shows that the minimum follow-up
times estimated by the proposed framework are broadly consistent with
evidence from long-term prostate cancer studies. Although these studies
differ in design and clinical objectives, they all indicate that clinically
relevant outcomes continue to emerge over extended follow-up,
supporting the more conservative follow-up times estimated by the
PDC and RSC.

\subsection{Triple-negative breast cancer data}
\label{subsec:mama_triplo_negativo}

The second application considers a dataset of patients diagnosed with
triple-negative breast cancer, originally analyzed by
\citet{milani2021analise}. The study included women treated at the
A.C.Camargo Cancer Center in São Paulo, Brazil, between 2001 and 2013
who underwent neoadjuvant chemotherapy. The event of interest was death
from breast cancer, whereas patients who died from other causes or did
not experience the event by the end of the observation period were
treated as censored observations.

The dataset comprises \(n=78\) patients, of whom \(25\) experienced the
event of interest and \(53\) were censored, yielding a censoring
proportion of \(67.95\%\). The observed follow-up times ranged from
\(9.00\) months to
\(\tmax=162.87\) months, and the latest observed event occurred at
\(\tilde{t}_{K}=60.77\) months.

We first evaluated follow-up sufficiency using the PFST and then fitted
the Weibull mixture cure model introduced in Section~3 to the data.
The maximum likelihood estimates were
\(\widehat{a}=2.2077\), \(\widehat{b}=32.4728\), and
\(\widehat{p}=0.6224\), corresponding to an estimated cure fraction of
\(62.24\%\). The fitted population survival function was

\[
\widehat{S}_p(t)
=
0.6224+
\left(1-0.6224\right)
\exp\left\{
-\left(\frac{t}{32.4728}\right)^{2.2077}
\right\}.
\]

Figure~\ref{fig:aplicacao_mama_triplo_negativo}
(\subref{fig:km_weibull_cura_mama_triplo_negativo}) shows the
Kaplan-Meier estimator as a black step curve and the population survival
curve estimated by the Weibull cure model as a solid red curve. The
estimated cure fraction, \(\widehat{p}=0.6224\), is represented by a
dashed blue horizontal line, whereas the terminal Kaplan-Meier estimate,
\(\widehat{p}_{KM}=0.6243\), is represented by a dotted gray horizontal
line. Figure~\ref{fig:aplicacao_mama_triplo_negativo}
(\subref{fig:hist_parametrico_mama_triplo_negativo}) shows the centered
bootstrap distribution used by the PFST together with its asymptotic
normal approximation. The gray histogram represents the centered
nonparametric bootstrap distribution, whereas the solid green curve represents
the asymptotic distribution obtained from the joint influence-function
approximation. The solid red vertical line indicates the observed PFST
statistic, whereas the solid black and blue vertical lines indicate the
bootstrap and asymptotic critical values, respectively.

\begin{figure}[!ht]
\centering
\begin{subfigure}[t]{0.49\textwidth}
\centering
\includegraphics[width=\textwidth]
{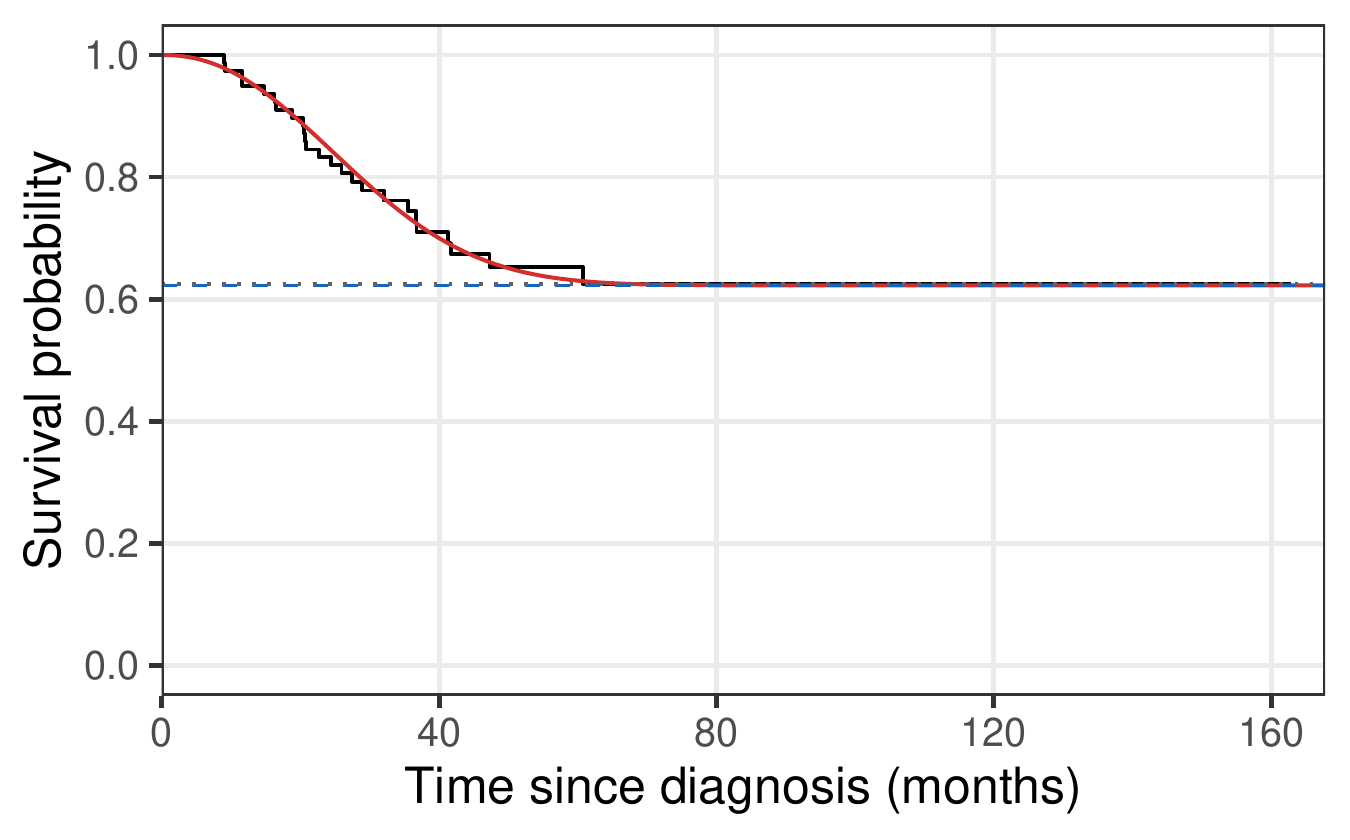}
\caption{Kaplan-Meier estimator and fitted Weibull cure model.}
\label{fig:km_weibull_cura_mama_triplo_negativo}
\end{subfigure}
\hfill
\begin{subfigure}[t]{0.49\textwidth}
\centering
\includegraphics[width=\textwidth]
{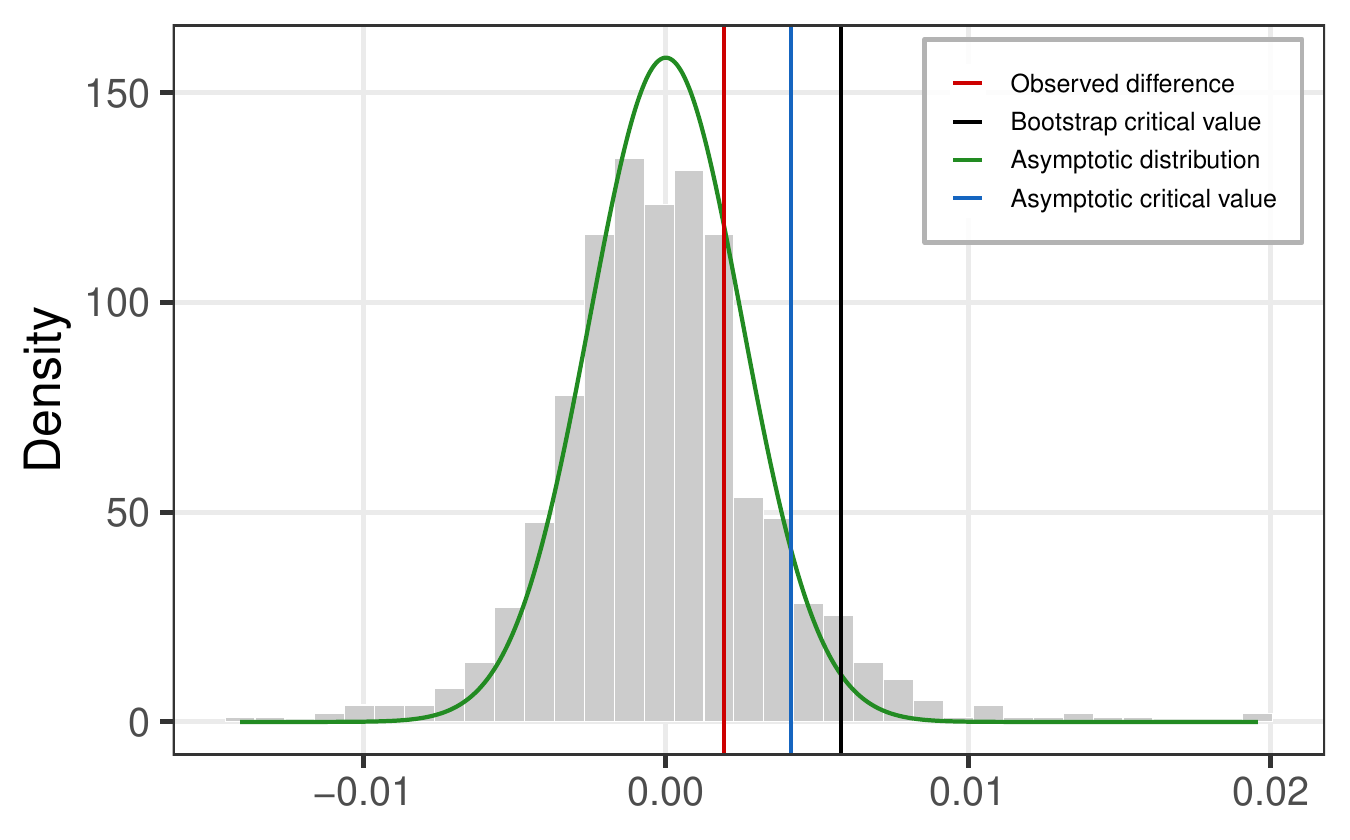}
\caption{Centered bootstrap distribution and asymptotic approximation of
the PFST statistic.}
\label{fig:hist_parametrico_mama_triplo_negativo}
\end{subfigure}
\caption{Graphical assessment of the Weibull cure model fit and the
PFST for the triple-negative breast cancer data.}
\label{fig:aplicacao_mama_triplo_negativo}
\end{figure}

The terminal Kaplan-Meier estimate was
\(\widehat{p}_{KM}=0.6243\), whereas the cure fraction estimated by the
model was \(\widehat{p}=0.6224\). Consequently, the observed PFST statistic was

\[
\Tparam
=
\widehat{p}_{KM}-\widehat{p}
=
0.0019.
\]

The bootstrap PFST decision was based on a centered nonparametric
bootstrap procedure with \(1000\) valid replications. The estimated
critical value at the \(5\%\) significance level was \(0.0058\), and the
bootstrap \(p\)-value was \(0.227\). The influence-function approximation
yielded \(\widehat{\tau}^{2}=0.000495\), corresponding to an estimated
standard error of \(0.0025\), an asymptotic critical value of \(0.0041\),
and an asymptotic \(p\)-value of \(0.223\). As shown in
Figure~\ref{fig:aplicacao_mama_triplo_negativo}
(\subref{fig:hist_parametrico_mama_triplo_negativo}), the red vertical
line corresponding to the observed PFST statistic lies to the left
of both the black bootstrap critical value and the blue asymptotic
critical value. Since \(0.0019<0.0058\) and \(0.0019<0.0041\), neither
procedure rejected \(H_0\). Thus, neither procedure provided evidence
that the available follow-up was insufficient for reliable estimation of
the cure fraction.

We subsequently applied the PDC and RSC to estimate the minimum
follow-up times associated with different tolerance levels.
Figure~\ref{fig:criterios_mama_triplo_negativo}
(\subref{fig:plateau_distance_mama_triplo_negativo}) shows the
Kaplan-Meier estimator as a black step curve, the fitted population
survival curve as a solid red curve, and the estimated cure fraction as
a dotted blue horizontal line. The colored dashed vertical lines represent
the minimum follow-up times obtained for
\(\delta\in\{0.100,\allowbreak 0.050,\allowbreak 0.025,\allowbreak 0.010,\allowbreak 0.005\}\).
Figure~\ref{fig:criterios_mama_triplo_negativo}
(\subref{fig:residual_survival_mama_triplo_negativo}) uses the same
graphical representation, with colored dashed vertical lines indicating
the minimum follow-up times obtained by the RSC for
\(\varepsilon\in\{0.100,\allowbreak 0.050,\allowbreak 0.025,\allowbreak 0.010,\allowbreak 0.005\}\).

\begin{figure}[!ht]
\centering
\begin{subfigure}[t]{0.49\textwidth}
\centering
\includegraphics[width=\textwidth]
{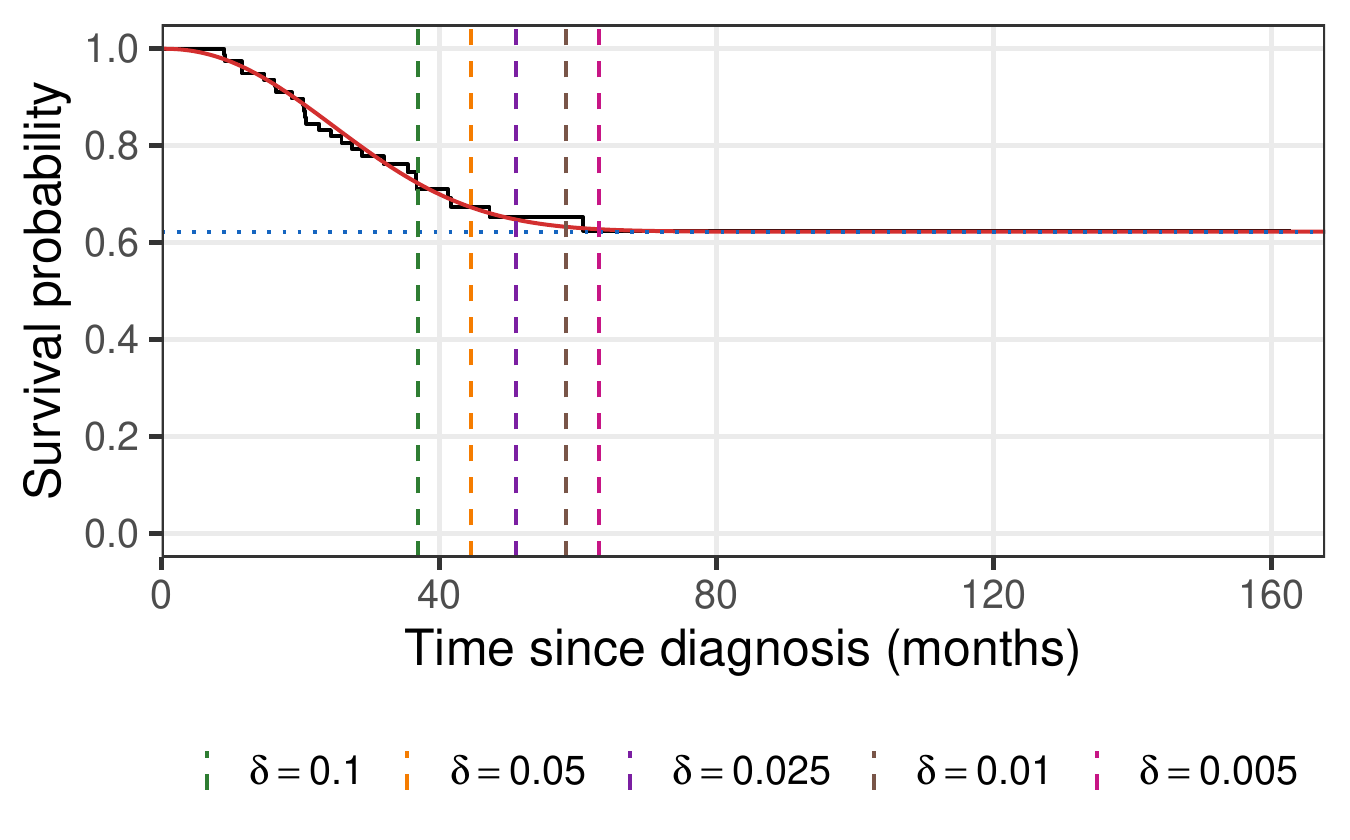}
\caption{Minimum follow-up times obtained by the PDC for different
tolerance levels \(\delta\).}
\label{fig:plateau_distance_mama_triplo_negativo}
\end{subfigure}
\hfill
\begin{subfigure}[t]{0.49\textwidth}
\centering
\includegraphics[width=\textwidth]
{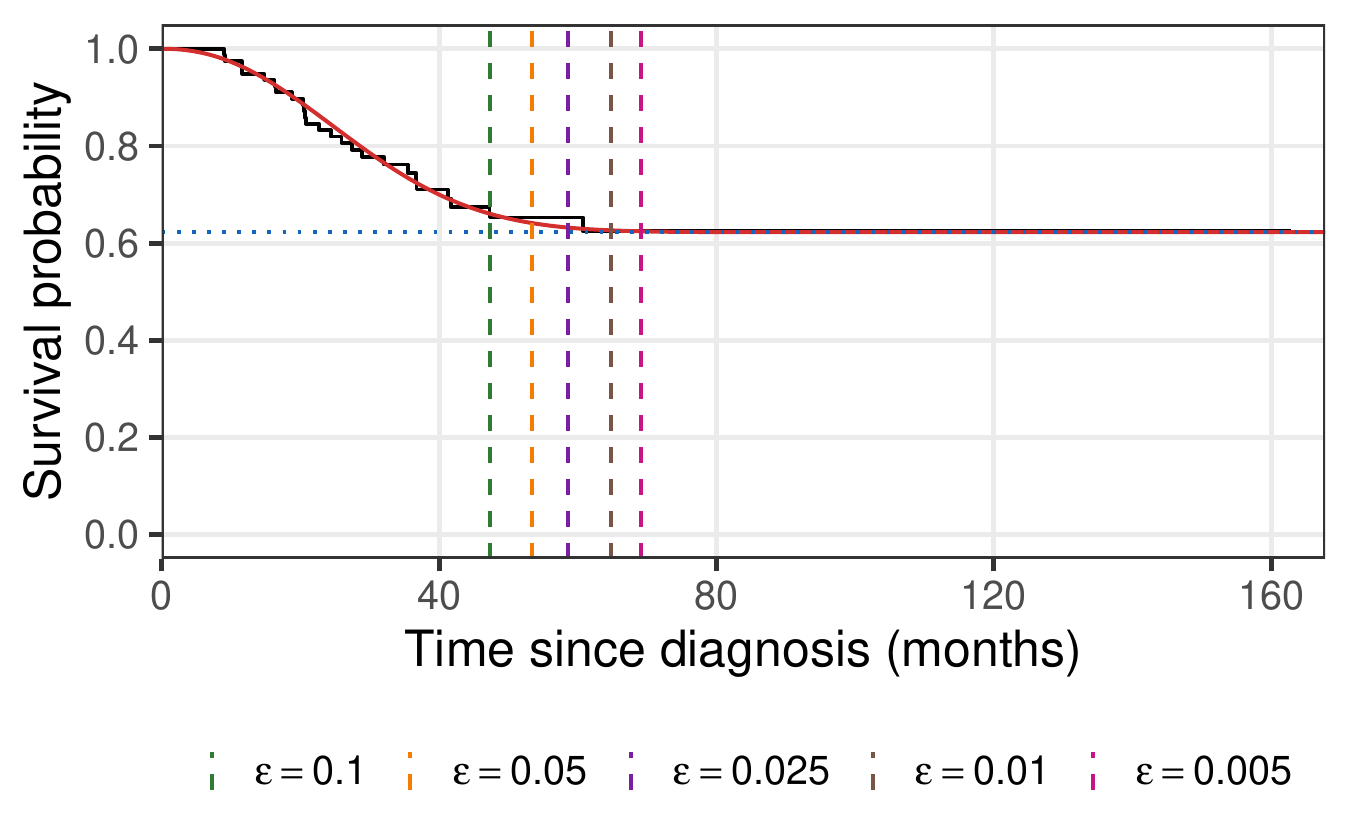}
\caption{Minimum follow-up times obtained by the RSC for different
tolerance levels \(\varepsilon\).}
\label{fig:residual_survival_mama_triplo_negativo}
\end{subfigure}
\caption{Minimum follow-up times estimated by the PDC and RSC for the
triple-negative breast cancer data. In both panels, the black step curve
represents the Kaplan-Meier estimator, the solid red curve represents
the fitted Weibull cure model, and the dotted blue horizontal line
represents the estimated cure fraction.}
\label{fig:criterios_mama_triplo_negativo}
\end{figure}

Table~\ref{tab:framework_summary_tnbc} summarizes the PFST,
PDC, and RSC results for the triple-negative breast cancer data. We
calculated the time difference as the estimated minimum follow-up minus
the maximum observed follow-up time of \(162.87\) months. Thus, negative
values indicate that the maximum observed follow-up exceeds the estimated
minimum, whereas positive values represent the additional follow-up time
required.

\begin{table}[!ht]
\centering
\caption{Summary of the follow-up sufficiency assessment for the
triple-negative breast cancer data.}
\label{tab:framework_summary_tnbc}
\renewcommand{\arraystretch}{1.20}
\small
\setlength{\tabcolsep}{8pt}

\begin{tabular}{llcc}
\toprule
\textbf{Method} &
\textbf{Tolerance} &
\makecell{\textbf{Observed or minimum}\\
\textbf{follow-up (months)}} &
\makecell{\textbf{Follow-up assessment}\\
\textbf{and time difference (months)}}\\
\midrule

PFST &
\makecell{Bootstrap: \(p=0.227\)\\Asymptotic: \(p=0.223\)} &
\makecell{\textit{Maximum observed}\\follow-up: 162.87} &
Sufficient\\

\midrule

PDC &
\(\delta=0.100\) &
36.93 &
Sufficient (\(-125.93\))\\

&
\(\delta=0.050\) &
44.67 &
Sufficient (\(-118.20\))\\

&
\(\delta=0.025\) &
51.05 &
Sufficient (\(-111.82\))\\

&
\(\delta=0.010\) &
58.24 &
Sufficient (\(-104.63\))\\

&
\(\delta=0.005\) &
63.03 &
Sufficient (\(-99.83\))\\

\midrule

RSC &
\(\varepsilon=0.100\) &
47.38 &
Sufficient (\(-115.49\))\\

&
\(\varepsilon=0.050\) &
53.38 &
Sufficient (\(-109.49\))\\

&
\(\varepsilon=0.025\) &
58.65 &
Sufficient (\(-104.21\))\\

&
\(\varepsilon=0.010\) &
64.85 &
Sufficient (\(-98.01\))\\

&
\(\varepsilon=0.005\) &
69.11 &
Sufficient (\(-93.76\))\\

\bottomrule
\end{tabular}
\end{table}

The PFST did not provide evidence of insufficient follow-up under either
the bootstrap or asymptotic procedure. Consistent with this result, the maximum observed
follow-up exceeded every minimum follow-up time estimated by the PDC and
RSC. Under the PDC, the maximum observed follow-up exceeded the estimated
minimum by \(99.83\)--\(125.93\) months. Under the RSC, it
exceeded the estimated minimum by \(93.76\)--\(115.49\)
months. Moreover, \(20\) patients had observed follow-up times at or
beyond the most stringent PDC estimate of \(63.03\) months, and \(15\)
patients had observed follow-up times at or beyond the most stringent RSC
estimate of \(69.11\) months. These late observations were censored.
Therefore, the available follow-up was sufficient under all tolerance
levels considered, although the maximum follow-up time should be
interpreted as the longest individual observation rather than as a common
follow-up duration for the entire cohort.

A direct comparison using identical numerical tolerance values
shows that the RSC produced longer minimum follow-up times than
the PDC. This difference is expected because equal numerical values of
\(\delta\) and \(\varepsilon\) represent different practical definitions
of follow-up sufficiency rather than equivalent levels of stringency.

In the application to the triple-negative breast cancer data,
\(\widehat{p}=0.6224\), so that \(1-\widehat{p}=0.3776\). The equivalent
tolerance pairs were
\((\delta,\varepsilon)=(0.100,0.2648)\),
\((0.050,0.1324)\),
\((0.025,0.0662)\),
\((0.010,0.0265)\), and
\((0.005,0.0132)\).
Substituting each pair into the expressions for
\(\widehat{t}_P(\delta)\) and \(\widehat{t}_R(\varepsilon)\) yielded the
same minimum follow-up times: \(36.93\), \(44.67\), \(51.05\), \(58.24\),
and \(63.03\) months, respectively. Consequently, the choice between the
two criteria should be guided by the practical interpretation of the
tolerance parameter rather than by numerical comparisons using identical
tolerance values.

In the SEER-based analysis by \citet{tai2005minimum}, breast cancer
required substantially longer follow-up, with a threshold time of
\(36.2\) years for patients younger than 60 years. This difference is
expected because the present application concerns a specific cohort of
patients with triple-negative breast cancer, whereas the analysis by
Tai et al.\ considered a broader breast cancer population.

Although this comparison establishes the methodological consistency of
the RSC with an existing criterion for minimum follow-up
assessment, it is also important to compare the estimated times with
clinical recurrence and mortality patterns reported specifically for
triple-negative breast cancer.
Table~\ref{tab:clinical_comparison_tnbc} summarizes this comparison.

\begin{table}[ht!]
\centering
\caption{Comparison of the minimum follow-up times estimated by the
proposed framework with clinical evidence for triple-negative breast
cancer.}
\label{tab:clinical_comparison_tnbc}
\small
\renewcommand{\arraystretch}{1.25}
\setlength{\tabcolsep}{5pt}

\begin{adjustbox}{max width=\textwidth}
\begin{tabular}{p{3cm}p{6.1cm}p{3cm}p{5.1cm}}
\toprule
\textbf{Study} &
\textbf{Main finding} &
\textbf{Clinical time horizon} &
\textbf{Consistency with the proposed framework}\\
\midrule

\citet{irvin2008triple}
&
The review reports that the risks of distant recurrence and death were
higher during the first five years after diagnosis, with distant
recurrence peaking at approximately three years.
&
Approximately 3--5 years
&
The three-year peak is close to
\(\widehat{t}_P(0.100)=36.93\) months, or approximately \(3.08\) years.
\\

\citet{chacon2010triple}
&
The peak recurrence rate occurred between the first and third years
after diagnosis, and most deaths occurred within the first five years.
&
Approximately 1--5 years
&
Consistent with
\(\widehat{t}_P(0.010)=58.24\) months and
\(\widehat{t}_P(0.005)=63.03\) months, corresponding to approximately
\(4.85\) and \(5.25\) years.
\\

\citet{zagami2022triple}
&
The risk of distant recurrence was substantially higher during the first
five years after diagnosis, whereas the risk of late recurrence after
five years was lower than \(3\%\).
&
Approximately 5 years
&
Consistent with
\(\widehat{t}_R(0.025)=58.65\) months, or approximately \(4.89\) years.
The more stringent RSC estimates ranged from approximately \(5.40\) to
\(5.76\) years.
\\

\bottomrule
\end{tabular}
\end{adjustbox}
\end{table}

Table~\ref{tab:clinical_comparison_tnbc} shows that the minimum follow-up times estimated by the proposed framework are broadly consistent with
the temporal pattern of recurrence and mortality reported for
triple-negative breast cancer. The clinical literature identifies the
first five years after diagnosis as the period of greatest risk, whereas
the PDC and RSC produce minimum follow-up times close to this
benchmark under moderate and stringent tolerance levels.

Taken together, the PDC, RSC, and clinical evidence suggest
that approximately five years of follow-up represents a clinically
relevant minimum horizon for assessing recurrence and mortality in
triple-negative breast cancer. More stringent control of the residual
survival probability among susceptible individuals may require
approximately \(5.4\) to \(5.8\) years of follow-up. The maximum observed follow-up was \(162.87\) months, or approximately \(13.57\) years, and there were \(15\) patients whose observed follow-up times were at or beyond the most stringent RSC estimate of \(69.11\) months. Thus, the upper tail of the observed follow-up distribution extended well beyond the estimated minimum requirements.

\section{Conclusion}

We developed a unified parametric framework for assessing follow-up sufficiency and estimating minimum follow-up requirements in mixture cure models. The framework combines the PFST, which provides an inferential assessment of the available follow-up, with the PDC and RSC, which express practical follow-up requirements on the population and susceptible survival scales, respectively.

The proposed criteria provide complementary perspectives on follow-up sufficiency. The PDC quantifies how closely the population survival curve approaches the cure plateau, making it particularly suitable when follow-up adequacy is considered from a population perspective. In contrast, the RSC controls the remaining survival probability among susceptible individuals and therefore directly assesses residual risk within the uncured subgroup. Although the two criteria are mathematically equivalent under an appropriate transformation of their tolerance parameters, they address different scientific objectives and should be viewed as complementary rather than competing approaches to follow-up assessment.

The simulation study showed that the PDC and RSC estimates become increasingly stable as the sample size and available observation period increase. Under joint estimation of the Weibull shape, scale, and cure-fraction parameters, the finite-sample precision of both criteria depends on the amount of information available about the susceptible survival distribution. The largest variability occurred under short follow-up, small sample sizes, and high cure fractions, whereas the estimates approached their theoretical targets in the more informative scenarios. The prostate cancer application illustrated a setting in which additional follow-up was required under moderate and stringent tolerances, whereas the triple-negative breast cancer application illustrated a setting in which the observed follow-up exceeded all the minimum requirements considered.

Unlike existing follow-up-sufficiency procedures that formulate the problem exclusively as a hypothesis test, the proposed framework integrates assessment and estimation. The PFST provides formal statistical evidence regarding follow-up adequacy, whereas the PDC and RSC translate this assessment into estimated minimum follow-up times that are directly interpretable for study planning and clinical research. Investigators can therefore select the criterion that best matches their scientific objective, depending on whether their interest lies in the convergence of the population survival curve to the cure plateau or in controlling the residual survival probability among susceptible individuals.

The tolerance parameters should not be interpreted as universal constants. \citet{boussari2018new} demonstrated that changing the conditional cure-probability threshold can materially alter the estimated time to cure, whereas \citet{dal2014long} obtained different time requirements depending on the future conditional relative-survival window and threshold considered. More generally, \citet{jakobsen2020estimating} showed that cure-point estimation depends jointly on the selected comparison measure, the margin of clinical relevance, and the estimation procedure. They also emphasized that confidence intervals may be wide when the comparison measure is nearly flat at the estimated threshold. The choice of \(\delta\) and \(\varepsilon\) should therefore be supported by subject-matter considerations, accompanied by sensitivity analyses, and interpreted together with the uncertainty in the estimated time requirements \citep{selukar2023receus}. Moreover, because follow-up times beyond the observed range are obtained through parametric extrapolation, misspecification of the cure fraction or susceptible survival distribution may affect the resulting estimates.

Although we illustrated the framework using the Weibull mixture cure model, it applies whenever the susceptible survival function is invertible. Future research may extend the PFST and the minimum follow-up calculations to flexible parametric and semiparametric cure models, regression settings, and covariate-dependent cure fractions. Recent approaches based on generalized additive models and neural networks offer one possible direction by relaxing restrictive assumptions about covariate effects \citep{dimari2025flexible}. A complementary extension would replace the single susceptible survival distribution with a finite mixture of latent fatal subgroups with distinct survival trajectories, as proposed by \citet{dimari2025modelbased}. In that setting, PDC and RSC could be defined either from the aggregate susceptible survival function or separately for clinically relevant latent subgroups.

\section*{Code availability}
The R scripts used to reproduce the simulation results, figures, and application analyses presented in this article are publicly available at \url{https://github.com/luizsilvaresende/pfst-cure-models}.

\bibliographystyle{apalike}
\bibliography{ref}

\end{document}